\documentclass[11pt]{article}
\usepackage[T1]{fontenc}

\usepackage{algorithm}
\usepackage{algpseudocode}
\usepackage{amsfonts}
\usepackage{amsmath}
\usepackage{amssymb}
\usepackage{amsthm}
\usepackage{authblk}
\usepackage{bbm}
\usepackage{enumerate}
\usepackage{graphicx}
\usepackage{ifthen}
\usepackage{latexsym}
\usepackage{mathpazo}
\usepackage{mathtools}
\usepackage{multicol}
\usepackage{palatino}
\usepackage{thmtools}

\usepackage[hidelinks]{hyperref}
\usepackage[capitalize]{cleveref}
\crefname{equation}{}{}

\hypersetup{
    pdftitle={Online learning of a panoply of quantum objects},
    pdfauthor={Akshay Bansal, Ian George, Soumik Ghosh, Jamie Sikora, Alice Zheng},
}

\hypersetup{pdfpagemode=UseNone}

\algnewcommand\algorithmicinput{\textbf{Input:}}
\algnewcommand\algorithmicoutput{\textbf{Output:}}
\algnewcommand\Input{\item[\algorithmicinput]}
\algnewcommand\Output{\item[\algorithmicoutput]}
\algnewcommand\To{\textbf{to}~}

\theoremstyle{definition}

\newtheorem{theorem}{Theorem}[section]
\newtheorem{lemma}[theorem]{Lemma}

\newtheorem{corollary}[theorem]{Corollary}
\newtheorem{definition}[theorem]{Definition}
\newtheorem{remark}[theorem]{Remark}

\newtheorem{proposition}[theorem]{Proposition}

\DeclarePairedDelimiter{\delarg}{(}{)}
\DeclarePairedDelimiterX{\delrel}[2]{(}{)}{#1\,\delimsize\|\,#2}
\DeclarePairedDelimiterX{\inner}[2]{\langle}{\rangle}{#1,\,#2}
\DeclarePairedDelimiter{\abs}{\lvert}{\rvert}
\DeclarePairedDelimiter{\norm}{\lVert}{\rVert}
\newcommand{\opnorm}[1]{\norm{#1}_{\text{op}}}
\newcommand{\trnorm}[1]{\norm{#1}_{\text{Tr}}}

\DeclareMathOperator*{\argmax}{arg\,max}
\DeclareMathOperator*{\argmin}{arg\,min}
\makeatletter
\newcommand{\newoperator}[6]{
    \expandafter\newcommand\csname #1@star\endcsname[#2][]{#3
        \if\relax\detokenize{#5}\relax\else
            \if\relax\detokenize{##1}\relax
                #4*#6
            \else
                #4*[##1]#6
            \fi
        \fi
    }
    \expandafter\newcommand\csname #1@nostar\endcsname[#2][]{#3
        \if\relax\detokenize{#5}\relax\else
            \if\relax\detokenize{##1}\relax
                #4#6
            \else
                #4[##1]#6
            \fi
        \fi
    }
    \expandafter\newcommand\csname #1\endcsname{\@ifstar
        {\csname #1@star\endcsname}
        {\csname #1@nostar\endcsname}}
}
\makeatother
\newoperator{ball}{2}{\operatorname{B}}{\delarg}{#2}{{#2}}
\newoperator{comm}{2}{\operatorname{comm}}{\delarg}{#2}{{#2}}
\newoperator{conv}{2}{\operatorname{conv}}{\delarg}{#2}{{#2}}
\newoperator{costrat}{2}{\operatorname{coStrat}}{\delarg}{#2}{{#2}}
\newoperator{density}{2}{\operatorname{D}}{\delarg}{#2}{{#2}}
\newoperator{effect}{2}{\operatorname{Eff}}{\delarg}{#2}{{#2}}
\newoperator{Herm}{2}{\operatorname{Herm}}{\delarg}{#2}{{#2}}
\newoperator{image}{2}{\operatorname{im}}{\delarg}{#2}{{#2}}
\newoperator{interior}{2}{\operatorname{int}}{\delarg}{#2}{{#2}}
\newoperator{linear}{2}{\operatorname{L}}{\delarg}{#2}{{#2}}
\newoperator{separable}{2}{\operatorname{SEP}}{\delarg}{#2}{{#2}}
\newoperator{spec}{2}{\operatorname{spec}}{\delarg}{#2}{{#2}}
\newoperator{supp}{2}{\operatorname{supp}}{\delarg}{#2}{{#2}}
\newoperator{Pd}{2}{\operatorname{Pd}}{\delarg}{#2}{{#2}}
\newoperator{Pos}{2}{\operatorname{Pos}}{\delarg}{#2}{{#2}}
\newoperator{strat}{2}{\operatorname{Strat}}{\delarg}{#2}{{#2}}
\newoperator{qgram}{2}{\operatorname{QGram}}{\delarg}{#2}{{#2}}
\newoperator{channel}{3}{\operatorname{C}}{\delarg}{#2}{{#2,#3}}
\newoperator{intmeas}{3}{\operatorname{IM}}{\delarg}{#2}{{#2,#3}}
\newoperator{F}{3}{\Phi_{#2}}{\delarg}{#3}{{#3}}
\newoperator{bregman}{4}{B_{#2}}{\delrel}{#3}{{#3}{#4}}
\newoperator{relentr}{3}{D}{\delrel}{#2}{{#2}{#3}}
\newoperator{reg}{2}{R}{\delarg}{#2}{{#2}}
\newoperator{entr}{2}{S}{\delarg}{#2}{{#2}}
\newoperator{trace}{2}{\operatorname{Tr}}{\delarg}{#2}{{#2}}
\newoperator{ptrace}{3}{\operatorname{Tr}_{#2}}{\delarg}{#3}{{#3}}

\newcommand{\regret}[1]{\mathcal{R}_{#1}}
\newcommand{\citep}{\cite}
\newcommand{\bigo}{\mathcal{O}}
\newcommand{\mbf}{\mathbf}
\newcommand{\mcl}{\mathcal}
\newcommand{\reals}{\mathbb{R}}
\newcommand{\complex}{\mathbb{C}}
\newcommand{\naturals}{\mathbb{N}}
\newcommand{\identity}{\mathbbm{1}}
\newcommand{\T}{^\mathrm{T}}
\newcommand{\adj}{^*}

\title{Online learning of a panoply of quantum objects}

\author[1]{Akshay Bansal}
\author[2]{Ian George}
\author[3]{Soumik Ghosh}
\author[1]{Jamie Sikora}
\author[1]{Alice Zheng\footnote{Corresponding author, \href{mailto:alicezheng@vt.edu}{\texttt{alicezheng@vt.edu}}.}}
\affil[1]{Department of Computer Science, Virginia Polytechnic Institute and State University}
\affil[2]{Department of Electrical and Computer Engineering, University of Illinois at Urbana-Champaign}
\affil[3]{Department of Computer Science, University of Chicago\vspace*{.5em}}

\date{October 7, 2024}

\begin{document}

\maketitle

\begin{abstract}
    In many quantum tasks, there is an unknown quantum object that one wishes to learn.
    An online strategy for this task involves adaptively refining a hypothesis to reproduce such an object or its measurement statistics.
    A common evaluation metric for such a strategy is its regret, or roughly the accumulated errors in hypothesis statistics.
    We prove a sublinear regret bound for learning over general subsets of positive semidefinite matrices via the regularized-follow-the-leader algorithm and apply it to various settings where one wishes to learn quantum objects.
    For concrete applications, we present a sublinear regret bound for learning quantum states, effects, channels, interactive measurements, strategies, co-strategies, and the collection of inner products of pure states.
    Our bound applies to many other quantum objects with compact, convex representations.
    In proving our regret bound, we establish various matrix analysis results useful in quantum information theory.
    This includes a generalization of Pinsker's inequality for arbitrary positive semidefinite operators with possibly different traces, which may be of independent interest and applicable to more general classes of divergences.
\end{abstract}

\section{Introduction}

Quantum tomography is arguably one of the most relevant learning problems in quantum information theory.
It is the problem of approximating the description of an unknown quantum state using one or more samples of the state \citep{cramer2010efficient, jullien2014quantum}, and has more recently found its relevance for predicting many other quantum objects.
In particular, tomography of quantum states is now well-understood in terms of its scaling with respect to the required number of samples.
A series of seminal results has shown that an unknown state can be estimated with sample complexity linear in the dimension of the state \citep{o2016efficient} and that this scaling is optimal \citep{haah2017sample}.
Note that the dimension of a quantum state is exponential in the number of qubits, so any algorithm for state tomography would exhibit this scaling.

Another prominent variant is the tomography of quantum channels, dubbed process tomography.
Given an unknown channel $\Phi$, the objective is to reconstruct it to within a certain margin of error.
Studying the number of interactive observables (measurement-state pairs) required to identify a quantum process, \cite{gutoski2014process} showed that an improvement from $\bigo(d^4)$ to $\bigo(d^2)$ observables is possible when the channel is \emph{a priori} known to be $d \times d$ unitary.
More recently, \cite{huang2023learning} considered a slight variant of the process tomography task where it is required to predict $\inner{E}{\Phi(\rho)}$ for any given observable $E$ and input state $\rho$.
Their work provided a learning algorithm that can predict the function value with high accuracy using a training dataset of size $\textup{poly}(d)$.

The approaches mentioned above share a common thread of learning the full representation of the object at hand, and all require at least a linear scaling in the dimension of the learned objects.
For quite some time, this has posed an existential question to the exponential number of amplitudes present within a quantum state -- for their description as such may not have been apt had an exponential number of such states been required to learn them.
Remarkably, for a slightly different task of \emph{shadow tomography} where one is required to learn just the outcome distribution for a given finite set of POVMs, the required number of copies of the unknown state has been shown to scale linearly in the number of qubits \cite{aaronson2018shadow}. The proof idea of this result essentially combines the two ideas of postselected learning of quantum states \cite{aaronson2004limitations} and the application of a gentle search procedure described in \cite{aaronson2005qma}.
The procedure for shadow tomography of quantum states has since been simplified via the use of classical shadows~\citep{huang_predicting_2020}.

One significant limitation of the traditional shadow tomography framework is its failure to address adversarial environments that evolve over time, a scenario more reflective of real-world modern experimental settings.
In a major result, \cite{aaronson2018online} demonstrated that quantum states can indeed be learned in such dynamic, online settings.
Building on this foundation, our work examines the learnability of various quantum objects -- such as quantum states, channels, POVMs, and strategies -- within an \emph{online} framework that accommodates adversarial interactions.

\subsection{Online learning and its application to quantum}\label{sec:previousWork}

Online learning encompasses a broad range of concepts and algorithms -- among them online gradient descent and multi-armed bandits.
Consider an interactive process between a learner and an adversary.
At each time point $t \in [T]$, the learner chooses an action $\omega_t \in \mcl{K}$ and the adversary responds with a loss function $f_t : \mcl{K} \to \reals$.
The learner incurs a penalty $f_t(\omega_t)$, which it aims to minimize.

The above framework gives a lot of power to the adversary.
To make it tractable, we typically require $f_t$ and $\mcl{K}$ to be bounded and convex \cite[Chapter~1]{hazan2016introduction}.
Even then, the adversary can take $\omega_t$ into account when choosing $f_t$, such as by designating each prediction as ``wrong'' unconditionally.
Computing the cumulative loss in such a scenario does not necessarily give a meaningful metric of success.
To account for the unavoidable amount of loss, we may compare against the best fixed action in hindsight via a metric called \emph{regret},
\begin{equation}\label{eq:total-regret}
    \regret{T} = \sum_{t = 1}^{T} f_t(\omega_t) - \min_{\varphi \in \mathcal{K}} \sum_{t = 1}^{T} f_t(\varphi)\ .
\end{equation}

The best action in hindsight $\varphi$ that minimizes the cumulative loss at time $T$ yields a possible choice for a subsequent hypothesis at time $T+1$.
Also known as \emph{follow-the-leader}, this rule occasionally leads to instability and hence suboptimal regret.
To mitigate this issue, the allowed actions can be coerced towards a common point using some convex \emph{regularizer} $\reg{}$.
Balancing the pull of this new function and the previous rule with a parameter $\eta > 0$ results in a new choice of hypothesis
\begin{equation}\label{eq:RFTL}
    \omega_{t+1} = \argmin_{\varphi \in \mathcal{K}} \left\{\eta \sum_{s = 1}^{t} f_s(\varphi) + \reg{\varphi} \right\}\ .
\end{equation}
Since $f_t$, $\reg{}$ and $\mcl{K}$ are convex, the process of choosing such a hypothesis via \cref{eq:RFTL} at each step reduces to a convex optimization problem.

Note lastly that instead of minimizing the loss functions directly, we may elect to use their linear tangent approximations.
This is additionally useful when each function $f_t$ is not revealed in its entirety, but instead just its subgradient at $\omega_t$.
We present the procedure discussed thus far as \cref{algo:generalRFTL}, formulated for learning over subsets of $\Herm{\mcl{X}}$, the set of complex Hermitian matrices acting on a vector space $\mcl{X}$.
\begin{algorithm}[H]
\caption{Regularized Follow-the-Leader (RFTL, a.k.a.~FTRL)}
\label{algo:generalRFTL}
\begin{algorithmic}[1]
    \Input $T$, $\eta > 0$, convex regularization function $\reg{}$, a convex and compact set $\mcl{K} \subseteq \Herm{\mcl{X}}$.
    \State Set initial hypothesis $\omega_1 \gets \argmin_{\varphi \in \mcl{K}} \{ \reg{\varphi} \}$.
    \For{$t \gets 1$ \To $T$}
        \State Predict $\omega_t$ and incur cost $f_t(\omega_t)$, where $f_t : \mcl{K} \to \reals$ is convex.
        \State Let $\nabla_t$ be a subgradient of $f_t$ at $\omega_t$ (assuming $f_t$ is such that a subgradient always exists).
        \State Update decision according to the RFTL rule
        \vspace*{-.75em}
        \begin{align}
            \omega_{t+1} \gets \argmin _{\varphi \in \mcl{K}}  \left\{ \eta\sum_{s = 1}^t \inner{\nabla_s}{\varphi} + \reg{\varphi} \right\}\ .
        \end{align}
        \vspace*{-1em}
    \EndFor
\end{algorithmic}
\end{algorithm}

An online algorithm that successfully learns the task at hand is bound to eventually decrease its incurred loss.
Specifically, learnability is exhibited by the average per-round regret $\regret{T}/T$ tending to $0$ as $T \to \infty$.
Any regret scaling less than $\mathcal{O}(T)$ guarantees this; as such, the first goal with any regret bound is to demonstrate its sublinearity with respect to $T$.

The convexity of $f_t$ in any online algorithm allows bounding its regret by $\inner{\nabla_t}{\omega_t - \varphi}$, with $\varphi$ as featured in $\cref{eq:total-regret}$.
As $\varphi$ is hard to analyze or compare to, we instead relate this to a term like $\inner{\nabla_t}{\omega_t - \omega_{t+1}}$ on a per-algorithm basis.
In the case of a regularized update rule, the introduced inaccuracies accumulate to an accompanying term related to the diameter of the object set $\mcl{K}$ with respect to the regularizer.
\begin{lemma}[\cite{hazan2016introduction} Lemma~5.3]\label{lemma:RegretBoundRFTL}
    \cref{algo:generalRFTL} guarantees
    \begin{align}
        \regret{T} \leq \sum_{t=1}^T \inner{\nabla_t}{\omega_t - \omega_{t+1}} + \frac{1}{\eta} D^2\ ,
    \end{align}
    where $D$ is the diameter of $\mcl{K}$ relative to the function $\reg{}$, i.e., $D^2 = \max_{\varphi, \varphi' \in \mcl{K}} \{\reg{\varphi} - \reg{\varphi'}\}$.
\end{lemma}
The inner product term may be interpreted as the ``stability'' of a particular algorithm; it is related to the distance between consecutive hypotheses, and more specifically upper-bounds the single-update objective difference $f_t(\omega_t) - f_t(\omega_{t+1})$.
This is also where the regularizer $\reg{}$ plays a role -- without a force to huddle the hypotheses, this difference may remain constant throughout all timesteps and thus lead to linear regret.
Instead, the inner product is typically bounded through requirements of additional properties, such as the regularizer being strongly convex.
For a more thorough discussion and analysis of relevant algorithms, see \cite{hazan2016introduction}.

\subsubsection{Online learning of quantum states and related work}\label{sec:relatedWork}

The framework of online learning may be applied to the task of learning quantum states through specific selections of loss functions.
In this task, we learn an unknown density operator $\rho \in \density{\mcl{X}}$, i.e., $\rho$ is positive semidefinite and has unit trace. We do so via its interactions with a sequence of measurements $0 \leq E_t \leq \identity$, i.e., $E_t$ and $\identity - E_t$ are positive semidefinite with $\identity$ being the identity matrix.
The probabilities of such measurements are $\inner{E_t}{\rho}$ when given $\rho$, and thus we define $f_t(X) = \ell_t(\inner{E_t}{X})$, with $\ell_t$ typically having relatively small values around $\inner{E_t}{\rho}$.
In the context of quantum tomography, such inner products are sometimes called \emph{shadows}.

The task of learning quantum states has seen extensive analysis in the framework of regret minimization.
Sublinear regret bounds for learning over the set of density operators $\mcl{K} = \density{\mcl{X}}$ are known for bandit feedback \citep{lumbreras2022multi}, and for full feedback via algorithms such as Exponentiated Gradient \citep{warmuth_online_2012}, Follow-the-Perturbed Leader \citep{yang_revisiting_2020}, Follow-the-Regularized Leader or Matrix Multiplicative Weights \citep{aaronson2018online}, and Volumetric-Barrier enhanced Follow-the-Regularized Leader \citep{jezequel_efficient_2022,tseng_online_2023}.
The relevant analyses require different assumptions on the loss functions, such as quadratic, $L_1$, $L$-Lipschitz, and logarithmic.
Notably, the latter case with VB-FTRL attains logarithmic regret, whereas in other cases it is square-root.
In a more recent work, \cite{chen2024adaptive} prove sublinear regret bounds for quantum states that evolve over time due to external factors such as measurements or environmental noise.
For more details about online learning and related algorithms, the reader is referred to \cite{hazan2016introduction,orabona_modern_2023}.

In this work, we focus primarily on the RFTL approach described by \cite{aaronson2018online}.
We may choose $\reg{X} = \trace{X \ln(X)}$ as the regularizer, which ensures the stability of the algorithm by coercing hypotheses to the maximally mixed state.
The RFTL update rule may be seen as minimizing a function $\F{E}{X} = \inner{E}{X} + \reg{X}$, with $E$ being the accumulated gradients.
When done over $\density{\mcl{X}}$, its optimum has the closed-form solution $e^{-E} / \trace{}(e^{-E})$.
This is known as a Gibbs state, which is well-studied in the context of quantum thermodynamics due to its relevance as the equilibrium state of a system described by the Hamiltonian $E$.
The theorems and proofs of \cite{aaronson2018online} for the large part mirror those of \cite{hazan2016introduction}, save for parts that rely on the particular choice of $\mcl{K}$.
In this work, we generalize these results to arbitrary compact and convex subsets of positive semidefinite matrices.

\subsubsection{Our setting}

The setting we consider in this work consists of a finite time-horizon partitioned into $T$ distinct time points.
At each time point $t \in [T]$, the adversary decides on a loss function $\ell_t : \reals \to \reals $ and a \emph{co-object} $E_t \in \mathcal{E}$.
The learner commits to an action $\omega_t \in \mcl{K}$ at time $t$ and then experiences the loss $\ell_t$ evaluated on the inner product $\inner{E_t}{\omega_t}$.
Subsequently, the entire loss function is revealed.
The notion of a co-object arises as a generalization of measurements in the case of learning over quantum states.
Such co-objects represent something that may meaningfully interact with the considered objects and produce some quantity of interest, i.e., their inner product.
Some examples of the relevant sets of co-objects are shown in \cref{sec:QIntroApps}.

We consider the set of allowed actions $\mcl{K}$ to be a subset of positive semidefinite operators $\Pos{\mcl{X}}$ and the set of co-objects $\mathcal{E}$ to be a subset of Hermitian operators $\Herm{\mcl{X}}$, which allows for capturing the interactions between various classes of quantum objects.
The restriction to positive semidefinite operators matches the setting of many analyses in quantum information.
We note that this is not a significant constraint in terms of applicability, with many quantum objects being subsets of positive semidefinite operators (such as quantum states, channels, POVMs, and strategies).
The objective of the learner is to devise a sequence of actions that minimizes their overall regret, as defined for general loss functions in \cref{eq:total-regret}.

\subsection{Main results}\label{sec:results}

We now present our main technical result, the proof of which can be found in \cref{sec:subLinRegret}.

\begin{theorem}[Sublinearity of regret]\label{thm:sublinearity}
    Suppose the loss function at each round $t \in \{1, \ldots, T\}$ is given by $f_t(\omega) = \ell_t(\inner{E_t}{\omega})$, where $f_t : \mcl{K} \to \reals$
    with $\mcl{K}$ compact and convex, and $\ell_t : \reals \to \reals$ is convex and $B$-Lipschitz.
    We additionally assume that there exists $\alpha > 0$ such $\alpha \identity \in \mcl{K}$.
    Then the bound on the regret $\regret{T}$ due to \cref{algo:qrftl} is given by
    \begin{equation}
        \regret{T} \leq 4 B C D \sqrt{A T}\ ,
    \end{equation}
    where $A$ is an upper bound on the trace of any object in $\mcl{K}$, $C$ is an upper bound on the operator norm of any co-object $E_t$ for $t \in \{ 1, \ldots, T\}$, and $D^2 = \max_{\varphi,\varphi' \in \mcl{K}} \{\entr{}(\varphi') - \entr{\varphi}\}$.
\end{theorem}

We discuss several quantum applications of the above result in \cref{sec:QIntroApps}.
While its proof is similar to other regret bound analyses, there are several important challenges faced in this context.

While one may wish to simply utilize the algorithm given by \cite{aaronson2018online} to learn an unknown quantum object, there are some inherent difficulties in establishing the learnability of a general quantum object (such as measurement operators, instruments, co-channels etc.) that do not have a predefined trace value.
We discuss some of these challenges below.

\paragraph{Challenge \#1: Variable traces.}
It is perhaps well-known that the regret incurred by regularized-follow-the-leader is bounded above as in \cref{lemma:RegretBoundRFTL}, which can be further bounded in terms of the trace distance between the successive positive semidefinite estimates $\omega_t$ and $\omega_{t+1}$ via H\"{o}lder's inequality.
When the estimates $\omega_t$ have equal trace values or the set $\mcl{K}$ consists of only objects with equal trace values, the trace distance could be subsequently bounded using the well-known quantum Pinsker's inequality \citep{ohya2004quantum}.
To bound the regret of positive semidefinite objects with possibly varying traces (such as for effects in a POVM), one may wish to develop a variant of Pinsker's inequality that holds for generic positive semidefinite objects.

To surmount this challenge, we offer two solutions.
One solution is to generalize Pinsker's inequality for our setting, as discussed below.
We briefly discuss another solution in \cref{sec:alternative}.

\paragraph{Solution \#1: A generalized Pinsker's inequality.}
We now state a generalized quantum Pinsker's inequality to include possibly unnormalized positive semidefinite objects.
Given the widespread use of Pinsker's inequality in the quantum literature, we believe that this version of Pinsker's inequality could be of independent interest and could find use in other applications.

\begin{theorem}[Generalized Pinsker's inequality]\label{thm:generalized-Pinsker's-inequality}
    For any positive semidefinite $P$ and $Q$, we have
    \begin{equation}\label{eq:generalized-Pinsker's-inequality}
        \frac{1}{4} \trnorm{P - Q}^2 \leq \max\{\trace{P}, \trace{Q}\} \left[ \relentr{P}{Q} - \trace{P - Q} \right]\ .
    \end{equation}
\end{theorem}

The non-triviality of this generalized statement stems from the inability to recover the additive $\trace{P-Q}$ term via the well-known constant-trace Pinsker's inequality, an approach seen in the proof of \cite[Theorem~10.8.1]{Wilde-Book} which only yields a quadratic term.
Our proof for this inequality extends upon the multivariate Taylor's theorem for the unnormalized variant of KL divergence and is detailed in \cref{sec:pinsker}.
We note that the techniques we develop to prove \cref{thm:generalized-Pinsker's-inequality} could also be extended to establish divergence inequalities for a more general class of Csisz\'ar $f$-divergences.
For instance, in a follow-up work by a subset of the authors, the integral representation of KL-divergence described in the proof of \cref{thm:classical-version} is extended to bound classical $f$-divergences in terms of $\chi^2$-divergence.
This results in novel contraction coefficient bounds as well as reverse Pinsker inequalities for classical $f$-divergences that differ from those in \cite{sason2016f} which relied on Lipschitz continuity of $f$.

Previously, numerical evidence suggested that the constant of $1/4$ in the previous equation can be tightened to $1/2$, which would recover quantum Pinsker's inequality for the case of equal traces of $P$ and $Q$.
One could also leverage strong convexity and the unit \emph{moduli of convexity} (with respect to the trace norm) of unnormalized von Neumann entropy shown in the unpublished notes of \cite{yu13} to obtain a tighter version of \cref{eq:generalized-Pinsker's-inequality}.
We remark that a generalized Pinsker's inequality for generic positive semidefinite objects is not explicitly present in the literature to the best of our knowledge.

\begin{remark}
    We note that \cref{thm:generalized-Pinsker's-inequality} is a strengthening of $\relentr{P}{Q} \geq \trace{P - Q}$, a special case of \emph{Klein's inequality} (also stated in \cref{fact:nonNegRelEnt}).
    For positive semidefinite $P$ and $Q$ (not both $0$), we obtain
    \begin{equation}
        \relentr{P}{Q} \geq \trace{P - Q} + \frac{1}{4 \max \{ \trace{P}, \trace{Q} \}} \trnorm{P - Q}^2\ ,
    \end{equation}
    which gives a strictly better lower bound when $P \neq Q$.
    A slightly stronger version of this inequality is given in \cref{cor:quantum-gen-pinsker}.
\end{remark}

\paragraph{Challenge \#2: Trace functional differentiation.}
Another obstacle in finding bounds on the regret for learning positive semidefinite objects with von Neumann entropy as the regularizer is extending gradients for trace functionals of the form $\trace{X \ln(X)}$, where $X$ is positive definite.
Specifically, the gradient of said functional over a subset of the positive definite cone can be used to derive an upper bound on the Bregman divergence, and subsequently on the overall regret.
The Fr\'{e}chet derivative of a different functional $-\ln(\det(X))$, which was examined by \cite{hjorungnes2007complex}, was recently utilized by \cite{zimmert2022pushing} to achieve logarithmic regret in a specific learning setting.
We note a purported absence of a similar discussion for the functional $\trace{X \ln(X)}$ in the existing literature.

\paragraph{Solution \#2: Fr\'{e}chet derivative.}
In this work, we also establish that the function
\begin{equation}
\Phi_{E}(X) \coloneqq \langle E , X \rangle + \trace{X \ln(X)}
\end{equation}
has a gradient for all $X \in \Pd{\mcl{X}}$ and that it is given by
\begin{align}
    \nabla \Phi_{E}(X) = E + \identity + \ln(X) \ .
\end{align}
In some sense, it would be surprising if the above claim were not true as, depending on how one defines $\nabla$ with respect to $X$, it is well-known $\nabla \inner{E}{X} = E$ and it would be intuitive that $\nabla \trace{X \ln(X)} = \identity + \ln(X)$ as $\frac{d}{dt}(t\ln t) = 1 + \ln(t)$.
However, we provide a proof of this result in \cref{sec:deriv-of-func-phi-E} since we could not find a complete proof of this result in the literature.
Also, we expect the tools and methodology used here to have applications beyond this work.
For example, in establishing this fact, we generalize a known trace functional result in \cite{Carlen-2010}.
Given the importance of trace functionals in quantum information theory (see \cite{Carlen-2010} for discussion), we expect this to be useful in other settings.

\medskip
Lastly, generalizing online learning to general quantum objects raises several pertinent considerations.
One is the question of choosing the regularizer, which needs to keep the iterates sufficiently close while maintaining a reasonable radius $D^2$.
As it turns out, the negative entropy function used by \cite{aaronson2018online} for learning quantum states still works.
We bound the distance between iterates via the generalized Pinsker's inequality (\cref{thm:generalized-Pinsker's-inequality}) and the radius via the following lemma.

\begin{lemma}\label{lem:dbound}
    Suppose $\mcl{K} \subseteq \Pos{\mcl{X}}$ satisfies $\trace{X} = A$ for all $X \in \mcl{K}$ and $A \geq 1$.
    Then we have
    \begin{equation}\label{eq:dbound-const}
        D^2 \leq A \ln(\dim(\mcl{X}))\ .
    \end{equation}
    Suppose $\mcl{K} \subseteq \Pos{\mcl{X}}$ satisfies $\trace{X} \leq A$ for all $X \in \mcl{K}$, and $A \geq 1$.
    Then we have
    \begin{equation}
        D^2 \leq \begin{cases}
            A \ln(\dim(\mcl{X})) & \text{ if } A \leq e^{-1} \dim(\mcl{X}) \\
            e^{-1} \dim(\mcl{X}) + A \ln(A) & \text{ if } A \geq e^{-1} \dim(\mcl{X})
        \end{cases}\ .
    \end{equation}
\end{lemma}

\subsubsection{Regret bounds for learning quantum objects}\label{sec:QIntroApps}

Using the bound above, we can prove sublinear regret bounds for learning many quantum objects.
We sketch how it can be applied to quantum channels and pure state inner products, then informally amalgamate the regret bounds for other objects.

\paragraph{Example \#1: Quantum channels.}
A \emph{quantum channel} is a physical operation that can be performed on quantum states (see \cref{sec:prelims} for definitions).
Suppose $\Phi$ is a quantum channel which maps quantum states acting on $\mcl{X}$ to quantum states acting on $\mcl{Y}$.
We can represent $\Phi$ by its Choi matrix $J$, and it just so happens that $\trace{J} = \dim(\mcl{X})$ for any choice of $J$.
Thus $A = \dim(\mcl{X})$ and we may now additionally make use of \cref{eq:dbound-const} in \cref{lem:dbound}.
The set of all such channels is a subset of $\Pos{\mcl{Y} \otimes \mcl{X}}$, and hence $D^2 \leq A \ln (\dim(\mcl{Y} \otimes \mcl{X})) = A \ln (\dim(\mcl{X})) + A \ln (\dim(\mcl{Y}))$.

It remains now to consider the set of \emph{co-objects} associated with quantum channels.
Much like we could interface with quantum states via measurements, we can interface with quantum channels via \emph{interactive measurements}.
Roughly, these describe the most general means of interacting with a given quantum channel.
Specifically, they describe preparing a larger state acting on $\mcl{X} \otimes \mcl{Z}$, sending the $\mcl{X}$ part through $\Phi$, and measuring the outcome.
An interactive measurement similarly has a Choi representation $R \in \Pos{\mcl{Y} \otimes \mcl{X}}$, meaning we may now take $\inner{J}{R}$ to obtain the effect observed when the interactive measurement acts on $\Phi$.
It happens that the constraints on $R$ ensure that $\opnorm{R} \leq 1$, which implies $C = 1$.
Substituting the values obtained for the constants $A$, $C$, and $D$ into the $4 BCD \sqrt{AT}$ regret bound of \cref{thm:sublinearity} yields the following.

\begin{lemma}
We have the following regret bound for learning quantum channels
    \begin{equation}
\regret{T} \leq \left( 4B \, \dim(\mcl{X}) \sqrt{\ln(\dim(\mcl{X})) + \ln(\dim(\mcl{Y}))} \right) \sqrt{T}\ .
    \end{equation}
\end{lemma}

We remark that this regret bound holds for any loss functions $\ell_t$ which are $B$-Lipschitz.

\paragraph{Example \#2: Collections of inner products of pure states.}
Suppose we are now given a collection of $n$ pure states $\{ |\psi_1 \rangle, \ldots, |\psi_n \rangle \}$ and are tasked to learn the \emph{Gram} matrix $G$ of their inner products defined by $G_{i,j} := \langle \psi_i | \psi_j \rangle$.
It necessarily follows that $G \in \Pos{\complex^n}$ and $G_{i,i} = 1$, which are also sufficient conditions for a collection of $n$ pure states to exist that matches such a matrix.
Note that this implies $\trace{G} = \sum_{i=1}^n G_{i,i} = n$, and so $A = n$.
We may again use \cref{eq:dbound-const} of \cref{lem:dbound} to obtain $D^2 \leq A \ln (\dim(\complex^n)) = A \ln (n)$.

The \emph{co-objects} in this setting depend on the application, but most generally are a bounded subset of $\Herm{\complex^n}$.
Specifically, we consider $E$ such that $\opnorm{E} \leq 1$, which implies $C = 1$.
Substituting into the regret bound of \cref{thm:sublinearity}, we get the following.

\begin{lemma}
    We have the following regret bound for learning the Gram matrix of $n$ pure states
    \begin{equation}
        \regret{T} \leq \left( 4 B n \sqrt{\ln(n)} \right) \sqrt{T}\ .
    \end{equation}
\end{lemma}

Again, this holds for general loss functions.

\paragraph{A panoply of other quantum examples.}

We informally state the results here, and refer to \cref{sec:quantumApplications} for added background and analyses on the discussed quantum objects.

\begin{theorem}[Informal, see \cref{sec:quantumApplications} for definitions and formal statements]
    There exists an online learning algorithm for the following sets of quantum objects:
    \begin{itemize}
        \item Quantum states (acting on $\mcl{X}$) with regret at most
        \begin{equation}
            \bigo \left( \sqrt{\ln(\dim(\mcl{X}))} \; \sqrt{T} \right);
        \end{equation}

        \item Quantum effects (acting on $\mcl{X}$) with regret at most
        \begin{equation}
            \bigo \left( \dim(\mcl{X}) \sqrt{\ln(\dim(\mcl{X}))} \; \sqrt{T} \right);
        \end{equation}

        \item Separable quantum states (acting on $\mcl{X}$) with regret at most
        \begin{equation}
            \bigo \left( \sqrt{\ln(\dim(\mcl{X})) } \; \sqrt{T} \right);
        \end{equation}

        \item The collection of inner products of $n$ pure states with regret at most
        \begin{equation}
            \bigo \left( n \sqrt{\ln(n) } \; \sqrt{T} \right);
        \end{equation}

        \item Quantum channels (with input $\mcl{X}$ and output $\mcl{Y}$) with regret at most
        \begin{equation}
            \bigo \left( \dim(\mcl{X}) \sqrt{\ln(\dim(\mcl{X})) + \ln(\dim(\mcl{Y})) } \; \sqrt{T} \right);
        \end{equation}

        \item Quantum interactive measurements (interacting with quantum channels as described above) with regret at most (assuming $\dim(\mcl{X}) \geq 3$, see \cref{sec:QCIM})
        \begin{equation}
            \bigo \left( \dim(\mcl{X}) \dim(\mcl{Y}) \sqrt{\ln(\dim(\mcl{X}))+\ln(\dim(\mcl{Y})) } \; \sqrt{T} \right);
        \end{equation}

        \item Quantum measuring strategies\footnote{Also called quantum combs \citep{chiribella_quantum_2008}.} (an object that takes the inputs $\mathcal{X}_1, \mathcal{X}_2, \ldots, \mathcal{X}_n$ and gives the outputs $\mathcal{Y}_1, \mathcal{Y}_2, \ldots, \mathcal{Y}_n$) with regret at most
        \begin{equation}
            \bigo \left(
                \left( \prod_{i=1}^n \dim(\mcl{Y}_i) \right)
                \left( \prod_{i=1}^n \dim(\mcl{X}_i) \right)
                \left( \sum_{i=1}^n \ln(\dim(\mcl{X}_i)) + \sum_{i=1}^n \ln(\dim(\mcl{Y}_i)) \right)^{1/2}
            \; \sqrt{T} \right) ;
        \end{equation}

        \item Quantum measuring co-strategies (interacting with quantum channels as described above) with regret at most
        \begin{equation}
            \bigo \left(
                \left( \prod_{i=1}^n \dim(\mcl{Y}_i) \right)
                \left( \prod_{i=1}^n \dim(\mcl{X}_i) \right)
                \left( \sum_{i=1}^n \ln(\dim(\mcl{X}_i)) + \sum_{i=1}^n \ln(\dim(\mcl{Y}_i)) \right)^{1/2}
            \; \sqrt{T } \right).
        \end{equation}
    \end{itemize}
\end{theorem}

\subsubsection{Limitations and future work}\label{sec:limitations}
A limitation of our bound is that it applies to general quantum objects, co-objects, and loss functions.
Obtaining regret bounds with regret strictly better than $\bigo(\sqrt{T})$ is possible in certain contexts.
For example, the VB-FTRL \citep{jezequel_efficient_2022} algorithm attains $\bigo(\ln (T))$ regret under the assumption of logarithmic losses, and has been extended to quantum states by \cite{tseng_online_2023}.
It might be possible to extend this work to obtain regret bounds closer to those in the above works by using specific loss functions and/or by taking advantage of the specific structure(s) of the set $\mcl{K} \subseteq \Pos{\mcl{X}}$.
We leave this as an interesting future research direction.

\subsection{Regret analysis}\label{sec:regret}

Applying the RFTL algorithm in the setting where $\mcl{K}$ is a generic compact subset of positive semidefinite matrices (denoted here as $\Pos{\mcl{X}}$) results in the following algorithm.

\begin{algorithm}[H]
\caption{RFTL for Online Learning of Quantum Objects}
\label{algo:qrftl}
\begin{algorithmic}[1]
    \Input $T$, $\eta > 0$, a convex and compact set $\mcl{K} \subseteq \Pos{\mcl{X}}$, and a bounded set $\mcl{E} \subseteq \Herm{\mcl{X}}$.
    \State Set initial hypothesis $\omega_1 \gets \argmax_{\varphi \in \mcl{K}} \{ \entr{\varphi} \}$.
    \For{$t \gets 1$ \To $T$}
        \State Predict $\omega_t$ and incur cost $f_t(\omega_t) \coloneqq \ell_t(\inner{E_t}{\omega_t})$ with $\ell_t : \reals \to \reals$ and $E_t \in \mcl{E}$.
        \State Let $\ell_t'(x)$ be a sub-derivative of $\ell_t$ with respect to $x$ and define
        \begin{equation}
            \nabla_t \gets \ell_t'(\inner{E_t}{\omega_t})E_t\ .
        \end{equation}
        \State Update decision according to the RFTL rule
        \begin{equation}\label{eq:learner-update-QRFTL}
            \omega_{t+1} \gets \argmin _{\varphi \in \mcl{K}}  \left\{ \eta\sum_{s = 1}^t \inner{\nabla_s}{\varphi} - \entr{\varphi} \right\}\ .
        \end{equation}
    \EndFor
\end{algorithmic}
\end{algorithm}
The regret bound for the general RFTL algorithm as given in \cref{lemma:RegretBoundRFTL} applies here also, its proof is essentially identical to that of \cite[Lemma~5.3]{hazan2016introduction}.
Note that we have further relaxed the assumption on the real-valued nature of $\mcl{K}$ which is now allowed to constitute any complex Hermitian object, the gradients at which can be analyzed using the techniques described by \cite{mordukhovich2006frechet}.
In compliance with \cref{algo:generalRFTL}, $\nabla_t$ is chosen to be the subgradient of our particular choice of $f_t$.
This can be shown using the chain rule for Fr\'{e}chet derivatives; see, for example, \cite{Coutts-2021a}.

The bound on the regret function given by \cref{lemma:RegretBoundRFTL} is often useful to assert learnability with RFTL.
To attain a regret bound sublinear in $T$, it suffices to show that $\inner{\nabla_t}{\omega_t - \omega_{t+1}}$ scales linearly with $\eta$ and then set $\eta \propto 1/\sqrt{T}$.
Our next lemma serves exactly this purpose.

\begin{lemma}\label{lemma:InnerProductBound}
    \cref{algo:qrftl} guarantees that for all $t \in \{1, \ldots, T-1\}$,
    \begin{equation}\label{eq:upperBoundInnerWithGrad}
        \inner{\nabla_t}{\omega_{t} - \omega_{t+1}} \leq 4 \eta \, \max\{\trace{\omega_t},\trace{\omega_{t+1}}\} \, \opnorm{\nabla_t}^2\ .
    \end{equation}
\end{lemma}

The detailed proof of \cref{lemma:InnerProductBound} makes use of several claims which we formalize in the technical sections of the paper and could be referred for details in \cref{sec:subLinRegret}.

We start by establishing in \cref{lemma:minimzerIsPD} that the learner's hypothesis $\omega_t$ given by \cref{eq:learner-update-QRFTL} in \cref{algo:qrftl} is positive definite under certain mild conditions on the set of allowed actions $\mcl{K}$.
We remark that the hypothesis $\omega_t$ is unique by strict convexity of the regularizer (see \cref{lemma:strictConvexityVonNeumann}), though this is not necessary for the proof.
In order to bound the inner product $\inner{\nabla_t}{\omega_t - \omega_{t+1}}$, it suffices to bound the trace norm $\trnorm{\omega_t - \omega_{t+1}}$ and $\opnorm{\nabla_t}$, as by H\"{o}lder's inequality we have $\inner{\nabla_t}{\omega_t - \omega_{t+1}} \leq \opnorm{\nabla_t} \trnorm{\omega_t - \omega_{t+1}}$.
To provide an upper bound on $\opnorm{\nabla_t}$, we can further decompose it as $\opnorm{\nabla_t} = \opnorm{\ell'_t(\inner{E_t}{\omega_t})} \opnorm{E_t}$ where $\ell'_t(\inner{E_t}{\omega_t})$ is the subgradient of the loss function $\ell_t$ at $\inner{E_t}{\omega_t}$ and $E_t$ is the corresponding co-object.
When $\ell_t$ is $B$-Lipschitz, it is well-known that $\opnorm{\ell'_t(\inner{E_t}{\omega_t})} \leq B$, for which we give an alternative proof in \cref{lemma:boundOnSubgrad}.
However, a bound on the operator norm of different types of co-objects (such as measurements, co-channels, co-strategies etc.) is a much lengthier discussion and is elaborated in parts of \cref{sec:quantumApplications}.
It remains to bound $\trnorm{\omega_t - \omega_{t+1}}$, which we accomplish in two different steps.
We first generalize the well-known quantum Pinsker's inequality to include possibly unnormalized positive semidefinite objects as given by our generalized Pinsker's inequality in \cref{thm:generalized-Pinsker's-inequality}.
As a second step, we upper bound the Bregman divergence (associated with the objective function in \cref{eq:learner-update-QRFTL}) in terms of the inner product $\inner{\nabla_t}{\omega_t - \omega_{t+1}}$ where we recall that the associated function needs to differentiable at $\omega_{t+1}$ for the Bregman divergence to be well-defined.
To this end, we invoke the idea of (Fr\'{e}chet) differentiability of a function at a point in the interior of its domain to calculate the derivative.
Note that the derivative is well-defined as $\omega_{t+1}$ is positive definite.
We combine these results by noting that the upper bound on $\trnorm{\omega_t - \omega_{t+1}}^2$ from \cref{thm:generalized-Pinsker's-inequality} is simply a scaled representation of Bregman divergence between $\omega_t$ and $\omega_{t+1}$ with the gradient of its associated function given by \cref{lem:OpDerivative}.
This provides an upper bound on $\trnorm{\omega_t - \omega_{t+1}}$ and thus results in an upper bound on $\inner{\nabla_t}{\omega_{t} - \omega_{t+1}}$ as given by \cref{lemma:InnerProductBound}.

We focus on the learning framework where the trace of a quantum object is always bounded.
The diameter of the object set $\mcl{K}$ with respect to the von Neumann entropy featured as the constant $D$ in the upper bound of the regret function in \cref{lemma:RegretBoundRFTL} can also be shown to be bounded using \cref{lem:dbound} (with proof details in \cref{sec:diameter}).
Finally, we bound the regret function in \cref{lemma:RegretBoundRFTL} using \cref{lemma:InnerProductBound}, and by setting $\eta$ proportional to $1/\sqrt{T}$, we get the bound in \cref{thm:sublinearity}.

\subsubsection{Alternative approach to attain sublinear regret} \label{sec:alternative}

We additionally present an alternative approach for attaining sublinear regret bounds when learning over sets of objects with unequal traces.
Given a convex and compact set $\mathcal{K} \subseteq \Pos{\mcl{X}}$ with the trace of its objects bounded by $A$, we can define the set
\begin{equation}\label{eq:bar-k}
    \overline{\mathcal{K}} = \left\{ \begin{bmatrix} X & 0 \\ 0 & A - \trace{X} \end{bmatrix} : X \in \mathcal{K} \right\}\ ,
\end{equation}
and note that all objects in $\overline{\mathcal{K}}$ have trace equal to $A$.
We similarly define
\begin{equation}
    \overline{\mathcal{E}} = \left\{ \begin{bmatrix} E & 0 \\ 0 & 0 \end{bmatrix} : E \in \mathcal{E} \right\}\ ,
\end{equation}
and note that there exist bijections between the original and modified sets.
We may thus run the learning algorithm over $\overline{\mathcal{K}}$ and $\overline{\mathcal{E}}$, while interfacing with the adversary over $\mathcal{K}$ and $\mathcal{E}$.
Given an additional assumption that $\mathcal{K} = \{ X \in \Pos{\mcl{X}} : \trace{X} \leq A \}$, the algorithm presented by \cite{aaronson2018online} yields a sublinear regret bound via a slight modification of the above setting.
For general $\mathcal{K}$, the proof method presented in \cref{sec:subLinRegret} (and other results therein) applied to the modified setting attains a sublinear regret bound, with the additional constant trace assumption allowing the use of Pinsker's inequality in place of \cref{thm:generalized-Pinsker's-inequality}.
This leads to the regret being bounded by $2 B C \overline{D} \sqrt{A T}$, where $\overline{D}$ is the diameter of $\overline{\mathcal{K}}$.
As $\overline{D}$ in general differs from $D$, the use of either regret bound may be beneficial in select circumstances.

\section{Preliminaries}\label{sec:prelims}

\subsection{Linear algebra notation and terminology}
In this section, we establish the notation and terminologies that we use in the rest of the paper.
We use calligraphic symbols like $\mathcal{X}$ and $\mathcal{Y}$ to refer to complex finite-dimensional Euclidean vector spaces.
We use the notation $\linear{\mcl{X}}$, $\Herm{\mcl{X}}$, $\Pos{\mcl{X}}$, and $\Pd{\mcl{X}}$ to denote the set of all linear operators, Hermitian operators, positive semidefinite operators, and positive definite operators acting on $\mcl{X}$, respectively.

We often use linear maps of the form $\Phi : \linear{\mcl{X}} \to \linear{\mcl{Y}}$.
We call $\Phi$ \emph{positive} if it holds that $\Phi(P) \in \Pos{\mcl{Y}}$ for every $P \in \Pos{\mcl{X}}$.
A map $\Phi$ is \emph{completely positive} if it holds that $\Phi \otimes \identity_{\linear{\mcl{Z}}}$, where $\identity_{\linear{\mcl{Z}}}$ is the identity map, is a positive map for every vector space $\mcl{Z}$.
A map is called \emph{trace-preserving} if it holds that $\trace{\Phi(X)} = \trace{X}$, for all $X \in \linear{X}$.
We denote the set of all completely positive and trace-preserving maps as $\channel{\mcl{X}}{\mcl{Y}}$.

We use various vector/matrix norms, depending on our needs.
The Euclidean norm of a vector $v \in \mcl{X}$ is given by
\begin{equation}
    \norm{v} \coloneqq \sqrt{\inner{v}{v}}\ .
\end{equation}
The trace of an operator $A \in \linear{\mathcal{X}}$ is given as $\trace{A}$.
The trace norm of an operator $X \in \linear{\mathcal{X}}$ is given by
\begin{equation}
    \trnorm{A} \coloneqq \trace{\sqrt{A\adj A}}\ .
\end{equation}
The operator norm of an operator $X \in \linear{\mcl{X}}$ is given by
\begin{equation}
    \opnorm{A} \coloneqq \max \left\{ \norm{Au} : u \in \mcl{X}, \norm{u} \leq 1 \right\}\ .
\end{equation}

    \subsection{Divergence-induced quantities and their properties}\label{definitions}

    Divergences and the quantities they induce are central to many statistical problems.
    As we see in the subsequent section of the background, online learning is no different in this respect.
    However, as this work considers general positive semidefinite objects rather than (normalized) probability measures as is standard, it is important to establish how we define these quantities in this work.
    We first address this.

    Throughout this work two types of Bregman divergence \citep{Bregman-1967a} are central.
    Bregman divergence has seen a great deal of use in classical \citep{Censor-1992a,Kiwiel-1997,Banerjee-2005a,hazan2016introduction}
    and more recently quantum \citep{Petz-2007a,Quadeer-2019a,He-2023a,Hayashi-2023a} contexts.
    We take the following definition.
    \begin{definition}
        Let $\mcl{K} \subset \mcl{X}$ be convex and $F:\mcl{K} \to \reals$ be convex on $\mcl{K}$ and differentiable on $\interior{\mcl{K}}$.
        The Bregman divergence of $F$ is defined for all $P \in \mcl{K}$ and $Q \in \interior{\mcl{K}}$ as
        \begin{equation}
            \bregman{F}{P}{Q} \coloneqq F(P) - F(Q) + \inner{\nabla F(Q)}{Q - P}\ .
        \end{equation}
    \end{definition}
    Note also that the standard definition requires $F$ to instead be strictly convex on $\mcl{K}$ in order for $\bregman{F}{}{}$ to be a metric.
    We remark that Bregman divergence is in effect the first-order Taylor expansion about $Q$ evaluated at point $P$, which is a useful intuition for later.

    For $P,Q \in \Pos{\mcl{X}}$, we define the quantum relative entropy as
    \begin{align}\label{eq:rel-ent-def}
        \relentr{P}{Q} &\coloneqq \begin{cases}
            \trace{P \ln(P) - P \ln(Q)} & \text{if}\ \image{P} \subseteq \image{Q}\\
            \infty & \text{otherwise}
        \end{cases}\ .
    \end{align}
    This quantity satisfies Klein's inequality \citep{klein1931quantenmechanischen,ruskai_inequalities_2002}, the proof of which can be found in \cite[Lemma~1]{lanford1968mean} and \cite[Proposition~5.22]{watrous2018theory}.
    \begin{remark}[Klein's inequality]\label{fact:nonNegRelEnt}
        Let $P,Q \in \Pos{\mcl{X}}$.
        Then, $\relentr{P}{Q} \geq \trace{P - Q}$, with equality if and only if $P = Q$.
    \end{remark}
    We define the quantum (a.k.a.~von Neumann) entropy for $X \in \Pos{\mcl{X}}$ by the relative entropy.
    \begin{equation}\label{eq:entropy-def}
        \entr{X} \coloneqq -\relentr{X}{\identity} = -\trace{X \ln(X)}\ .
    \end{equation}
    This quantity satisfies the following bounds for $\rho \in \density{\mcl{X}}$.
    \begin{equation}\label{eq:entropy-bounds}
        0 \leq S(\rho) \leq \ln(\dim(\mcl{X})) \ ,
    \end{equation}
    where the lower bound is saturated by $\rho = vv^{\ast}$ for some unit vector $v$, i.e.~$\|v\|_{2} = 1$, and the upper bound is saturated for $\rho = \dim(\mcl{X})^{-1} \identity_{\mcl{X}}$.
    We lastly note the following property.
    \begin{remark}\label{lemma:strictConvexityVonNeumann}
        Von Neumann entropy $\entr{}$ is strictly concave on $\Pos{\mcl{X}}$.
    \end{remark}
    There are various means of proving this, among them the strict concavity of Shannon entropy \cite{wehrl_general_1978}, or \cite[Theorem~5.23]{watrous2018theory} together with the equality condition of Klein's inequality.

    We remark upon two functionals that are used throughout this work that are induced by the negative relative entropy, $R(X) \coloneqq -S(X)$.
    The first is
    \begin{equation}\label{eq:Phi-E-defn}
        \F{E}{X} \coloneqq \inner{E}{X} + \reg{X}\ ,
    \end{equation}
    which is well-defined for any $E \in \Herm{\mcl{X}}$ and $X \in \Pos{\mcl{X}}$.
    The second is the Bregman divergence defined by the regularizer $\reg{}$,
    \begin{equation}\label{eq:neg-ent-breg-defn}
        \bregman{\reg{}}{P}{Q} \coloneqq \bregman{-\entr{}}{P}{Q} = \relentr{P}{Q} - \trace{P - Q}\ ,
    \end{equation}
    where this equality may be established by using the results in \cite{Carlen-2010} or \cref{sec:deriv-of-func-phi-E}.
    We may similarly define a Bregman divergence via $\Phi_E$ for any $E \in \Herm{\mcl{X}}$,
    \begin{align}\label{eq:breg-phi-e-defn}
        \bregman{\Phi_E}{P}{Q} = \Phi_{E}(P) -\Phi_{E}(Q) + \inner{\nabla \Phi_E (Q)}{Q - P}\ ,
    \end{align}
    which is well-defined for any $P \in \Pos{\mcl{X}}$ and $Q \in \Pd{\mcl{X}}$.
    This simplifies to
    \begin{equation}
        \bregman{\Phi_E}{P}{Q} = \relentr{P}{Q} + \trace{P - Q} = \bregman{-\entr{}}{P}{Q}\ ,
    \end{equation}
    which follows from \cref{eq:rel-ent-def}, \cref{eq:entropy-def}, and following lemma.
    \begin{lemma}\label{lem:OpDerivative}
        Let $E \in \Herm{\mcl{X}}$ and $X \in \Pd{\mcl{X}}$.
        The gradient of $\Phi_{E}(X) \coloneqq \langle E , X \rangle + \trace{X \ln (X)}$ is
        \begin{equation}
            \nabla \Phi_E(X) = E + \identity_{\mathcal{X}} + \ln(X) \ .
        \end{equation}
    \end{lemma}
    While we believe the above lemma is to be expected, to the best of our knowledge, there is not a complete proof in the literature.
    As such, we establish this lemma in \cref{sec:deriv-of-func-phi-E}.

\section{Quantum applications}\label{sec:quantumApplications}

In this section, we apply our general regret bound for positive semidefinite objects in the landscape of quantum information.
Quantum objects often come in pairs, a physical object and what we refer to here as a co-object.
The co-object is another physical thing which \emph{interacts} with the object to create an outcome.
This concept is best illustrated with the example of measuring a quantum state, which we discuss shortly.

This section is organized as follows.
We introduce a quantum object and its positive semidefinite representation (which always corresponds to a convex, compact set and includes $\alpha \identity$ for some $\alpha > 0$) then continue to upper bound its largest trace $A$ and its operator norm $C$ so that we can apply our regret bound.

\subsection{Learning quantum states and measurements}

A quantum state is a physical object that can be described mathematically by a Hermitian positive semidefinite matrix $\rho$ with trace $1$.
We call such matrices \emph{density operators} and denote them by
\begin{equation}
    \density{\mcl{X}} \coloneqq \{ \rho \geq 0: \trace{\rho} = 1 \}\ ,
\end{equation}
where $\mcl{X}$ is the complex, finite-dimensional vector space that $\rho$ acts on.
If $\rho = vv\adj$ for some unit vector $v$, we call $\rho$ a \emph{pure state}.

\begin{itemize}
    \item When $\mcl{K} = \density{\mcl{X}}$, we have $A = 1$, by definition.
    \item When $\mcl{E} = \density{\mcl{X}}$, we have $C = 1$, attained at any pure state.
\end{itemize}

When you measure a quantum state, you observe an \emph{effect}, which occurs with some probability.
An effect is represented by a positive semidefinite operator $E$.
Born's rule states that the probability of observing \emph{E} while measuring a quantum state $\rho$ occurs with probability $\inner{E}{\rho}$.
Mathematically, the set of effects is denoted by
\begin{equation}
    \effect{\mcl{X}} \coloneqq \{ E \geq 0: E \leq \identity_\mcl{X} \}\ ,
\end{equation}
which can be checked to be the set of all operators that yield proper probabilities given by Born's rule.

\begin{itemize}
    \item When $\mcl{K} = \effect{\mcl{X}}$, we have $A = \dim(\mcl{X})$, attained at $E = \identity_{\mcl{X}}$.
    \item When $\mcl{E} = \effect{\mcl{X}}$, we have $C = 1$, since every effect has eigenvalues between $0$ and $1$.
\end{itemize}

We now apply our regret bound to the online learning of quantum states.
Here, the quantum objects are quantum states and the co-objects are quantum effects.
As mentioned previously, this setting was studied by \cite{aaronson2018online} and we recover their regret bound (up to a constant\footnote{Note that we lose a constant factor since our generalized Pinsker's inequality has a slightly smaller constant in the variable-trace regime as compared to the standard Pinsker's inequality.}), below.

\begin{corollary}[Online learning of quantum states, see also \cite{aaronson2018online}]
\label{app:states}
For $\mcl{K} = \density{\mcl{X}}$ and for $\mcl{E} = \effect{\mcl{X}}$, we have
    \begin{equation}
        \regret{T} \leq \left( 4 B \sqrt{ \ln(\dim(\mcl{X}))} \right) \sqrt{T}\ .
    \end{equation}
\end{corollary}

Note that we can immediately apply our result in the other direction as well.
If the object that we are trying to learn is an effect itself, then the co-object is a quantum state.
Note that effects are not constant-trace sets, hence our generalized Pinsker's inequality is key to unlocking this result.

\begin{corollary}[Online learning of quantum effects]
\label{app:effects}
For $\mcl{K} = \effect{\mcl{X}}$ and for $\mcl{E} = \density{\mcl{X}}$, we have
    \begin{equation}
        \regret{T} \leq \left( 4B \, {\dim(\mcl{X})} \sqrt{e^{-1} + \ln(\dim(\mcl{X}))} \right) \sqrt{T}\ .
    \end{equation}
\end{corollary}

One of the most examined subsets of quantum states are those that are entangled, and by complementarity, not entangled.
To consider this setting, we must have that the vector space $\mcl{X}$ decomposes into the tensor product $\mcl{X}_1 \otimes \mcl{X}_2$.
The set of states that are not entangled are called separable, and are defined as
\begin{equation}
    \separable{\mcl{X}_1:\mcl{X}_2} \coloneqq \conv{} \{ \rho_1 \otimes \rho_2 :
    \rho_1 \in \density{\mcl{X}_1} \text{ and }
    \rho_2 \in \density{\mcl{X}_2} \}\ ,
\end{equation}
where $\conv{}$ denotes the convex hull.

If we wish to consider learning separable states, there are two ways to choose co-objects.
One may consider co-objects to be so-called \emph{entanglement witnesses}, i.e., the most general co-object that yields a proper probability under Born's rule.
While our bound applies in this setting, entanglement witnesses do not have any physical interpretation (as far as we are aware).
However, we can still choose the co-objects as effects, as we do below.

\begin{corollary}[Online learning of separable states]
\label{app:sep}
For $\mcl{K} = \separable{\mcl{X}_1:\mcl{X}_2}$ and $\mcl{E} = \effect{\mcl{X}_1 \otimes \mcl{X}_2}$, we have
    \begin{equation}
\regret{T}
\leq \left( 4 B \sqrt{ \ln(\dim(\mcl{X}))} \right) \sqrt{T}\ .
    \end{equation}
\end{corollary}

Note that while the above regret bound is identical to that in  \cref{app:states}, the RFTL algorithm does make use of the added structure.
In particular, at any point in time, the hypothesis state $\omega_t$ is a separable state.
If we ran RFTL with $\mcl{K} = \density{\mcl{X}_1 \otimes \mcl{X}_2}$ in this setting, it could be the case that $\omega_t$ is not separable at every time step.
This could be important in certain contexts where one wishes to maintain separability.

\subsection{Learning a collection of inner products of pure states}

Suppose now that the object to be learned is not just a single quantum state, but rather a collection of pure states.
Suppose we have $n$ pure states $\{ v_1 v_1\adj, \ldots, v_n v_n\adj \}$.
We denote the Gram matrix of these vectors as $G$ which is defined element-wise as
\begin{equation}
    G_{i,j} := \langle v_i, v_j \rangle\ ,
\end{equation}
i.e., the matrix of inner products.
A matrix is a Gram matrix if and only if it is positive semidefinite, and if those vectors all have unit norm, then $G$ also satisfies $G_{i,i} = 1$ for all $i$.
Thus, the set of Gram matrices corresponding to pure quantum states is a convex and compact set given by
\begin{equation}
    \qgram{n} := \{ G \in \Pos{\complex^n} : G_{1,1} = \cdots = G_{n,n} = 1 \}\ .
\end{equation}
In the online learning context with (quantum) Gram matrices, one is trying to learn the collection of pair-wise inner products.

The co-objects in this setting can vary depending on the application.
To keep it general, we use the set $\mcl{E} = \Herm{\complex^n}$ with bounded norm, i.e., $E \in \mcl{E}$ if and only if $\opnorm{E} \leq 1$.
This is the unit ball, and, for brevity, we denote this as
\begin{equation}
    \ball{n} := \{ E \in \Herm{\complex^n} : \opnorm{E} \leq 1 \}\ .
\end{equation}
In the online learning framework, the algorithm receives feedback depending on $\inner{G}{E}$.
Note that by bounding the norm of $E$, we do not get inner products that grow out of control needlessly.

\begin{itemize}
    \item When $\mcl{K} = \qgram{n}$, we have $A = n$, immediately from the definition.
    \item When $\mcl{E} = \ball{n}$, we have $C = 1$, by definition.
\end{itemize}

In this setting, we have the following regret bound.

\begin{corollary}[Online learning of inner products of $n$ pure states]
\label{app:qgram}
For $\mcl{K} = \qgram{n}$ and for $\mcl{E} = \ball{n}$, we have
    \begin{equation}
\regret{T}
\leq \left( 4 B n \sqrt{\ln(n)} \right) \sqrt{T}\ .
    \end{equation}
\end{corollary}

\subsection{Learning quantum channels and interactive measurements} \label{sec:QCIM}

A \emph{quantum channel} is a physical operation that can be performed on quantum states.
For this, we require these linear maps to be \emph{completely positive} and \emph{trace-preserving}, as defined in \cref{sec:prelims}.
Suppose $\Phi$ is a quantum channel which maps $\density{\mcl{X}}$ to $\density{\mcl{Y}}$ where $\mcl{Y}$ is another complex finite-dimensional vector space ($\dim(\mcl{X})$ and $\dim(\mcl{Y})$ need not be related).
We can represent $\Phi$ by its Choi matrix $J$ belonging to the following set
\begin{equation}
    \channel{\mcl{X}}{\mcl{Y}} \coloneqq \{ J \in \Pos{\mcl{Y} \otimes \mcl{X}} : \ptrace{\mcl{Y}}{J} = \identity_{\mcl{X}} \}\ ,
\end{equation}
where $\ptrace{\mcl{Y}}{}$ is the \emph{partial trace over $\mcl{Y}$} which is the unique linear map that satisfies the equation
\begin{equation}
    \ptrace{\mathcal{Y}}{X \otimes Y} = \trace{Y} \cdot X
\end{equation}
for all $X \in \Herm{\mcl{X}}$ and $Y \in \Herm{\mcl{Y}}$.

\begin{itemize}
    \item When $\mcl{K} = \channel{\mcl{X}}{\mcl{Y}}$, we have $A = \trace{\ptrace{\mcl{Y}}{J}} = \trace{J} = \trace{\identity_{\mcl{X}}} =  \dim(\mcl{X})$, since the partial trace is trace-preserving.
    \item When $\mcl{E} = \channel{\mcl{X}}{\mcl{Y}}$, we have $C = \opnorm{J} \leq \trnorm{J} = \trace{J} = A = \dim(\mcl{X})$ noting the trace norm is equal to the trace since it is positive semidefinite.
    This bound can be attained by any linear map which performs an isometry.
\end{itemize}

To study the online learning of quantum channels, we now discuss their so-called \emph{co-objects}.
An \emph{interactive measurement} is best described as the most general way to interact with a quantum channel.
Suppose one were to create a quantum state $\rho \in \density{\mcl{X} \otimes \mcl{Z}}$, where $\mcl{Z}$ is some complex finite-dimensional vector space corresponding to a \emph{memory}.
Then if we apply $\Phi$ to just the $\mcl{X}$ part of $\rho$, we would have a new quantum state $\rho' \in \density{\mcl{Y} \otimes \mcl{Z}}$.
If we measure $\rho'$, we observe the effect $E$ with some probability.
This entire process is called an interactive measurement and can be represented by its Choi matrix which belongs to the set
\begin{equation}
    \intmeas{\mcl{X}}{\mcl{Y}} \coloneqq \{ R \in \Pos{\mcl{Y} \otimes \mcl{X}} : R \leq \identity_{\mcl{Y}} \otimes \sigma, \, \sigma \in \density{\mcl{X}} \}\ .
\end{equation}

\begin{itemize}
    \item When $\mcl{K} = \intmeas{\mcl{X}}{\mcl{Y}}$, we have $A = \dim(\mcl{Y})$.
    \item When $\mcl{E} = \intmeas{\mcl{X}}{\mcl{Y}}$, we have $C = 1$.
    Both of these bounds are attained by any $R$ of the form $\identity_{\mcl{Y}} \otimes \sigma$.
\end{itemize}

The last ingredient for the online learning setting is that the probability of observing the effect associated with $R$ when the interactive measurement interacts with the channel $J$ is given by $\inner{J}{R}$, see \cite{Gutoski_2007}.
This brings us to the following regret bound.

\begin{corollary}[Online learning of quantum channels]
\label{app:channels}
For $\mcl{K} = \channel{\mcl{X}}{\mcl{Y}}$ and $\mcl{E} = \intmeas{\mcl{X}}{\mcl{Y}}$, we have
    \begin{equation}
        \regret{T} \leq \left( 4B \, \dim(\mcl{X}) \sqrt{\ln(\dim(\mcl{X})) + \ln(\dim(\mcl{Y}))} \right) \sqrt{T}\ .
    \end{equation}
\end{corollary}

Flipping the roles of channels and interactive measurements, we have the following.

\begin{corollary}[Online learning of interactive measurements]\label{app:interactive}
    \newcommand{\npt}{\hspace{-.2pt}}
    For $\mcl{K}\npt=\npt\intmeas{\mcl{X}}{\mcl{Y}}$ and $\mcl{E}\npt=\npt\channel{\mcl{X}}{\mcl{Y}}$, we have
    \begin{equation}
        \regret{T} \leq \begin{cases}
            \left( 8B \, \dim(\mcl{Y}) \sqrt{2/e + \ln(\dim(\mcl{Y}))} \right) \sqrt{T} & \text{ if } \dim(\mcl{X}) = 2 \\
            \left( 4B \, \dim(\mcl{X}) \dim(\mcl{Y}) \sqrt{\ln(\dim(\mcl{X}))+\ln(\dim(\mcl{Y}))} \right) \sqrt{T} & \text{ if } \dim(\mcl{X}) \geq 3
        \end{cases}
    \end{equation}
\end{corollary}

\subsection{Learning quantum interactions}

A quantum strategy is a prescribed interaction with a (compatible) object.
For example, suppose Alice and Bob communicate back and forth via a quantum communications network.
Imagine if Bob sends a quantum state in $\mcl{X}_1$ to Alice, who then sends a quantum state in $\mcl{Y}_1$ back to Bob.
Suppose they keep exchanging quantum states in this manner via the spaces $\mcl{X}_2, \mcl{Y}_2, \ldots, \mcl{X}_n, \mcl{Y}_n$ and Bob finally measures at the end to get an outcome.
Alice's actions are described by the $n$-turn strategy which is an object that takes the inputs $\mathcal{X}_1, \mathcal{X}_2, \ldots, \mathcal{X}_n$ and gives the outputs $\mathcal{Y}_1, \mathcal{Y}_2, \ldots, \mathcal{Y}_n$.
Bob's actions are described by the $n$-turn co-strategy which is the {(co-)object} that gives the inputs $\mathcal{X}_1, \mathcal{X}_2, \ldots, \mathcal{X}_n$ to Alice and takes the outputs $\mathcal{Y}_1, \mathcal{Y}_2, \ldots, \mathcal{Y}_n$ from Alice.
Technically, Alice and Bob can have memory spaces as well, but they do not factor into their Choi representations using the formalism in \cite{Gutoski_2007, gutoski2012quantum} which we use in this work.

We now examine the online learning of strategies.
Note that even though co-strategies are defined as interactions, they correspond to a single interaction in the online learning framework; every co-strategy corresponds to a single time step $t$.

We now define the set of Choi representations of strategies recursively as
\begin{subequations}
\begin{align}
    \strat{\mcl{X}_{1}, \ldots, \mcl{X}_n, \mcl{Y}_{1}, \ldots, \mcl{Y}_n} = \{ & Q \in \Pos{\mcl{Y}_{1} \otimes \cdots \otimes \mcl{Y}_{n} \otimes \mcl{X}_{1} \otimes \cdots \otimes \mcl{X}_n} : \\
        & \ptrace{\mcl{Y}_n}{Q} = Q' \otimes \identity_{\mcl{X}_n}, \; \text{ for some } \\
        & Q' \in \strat{\mcl{X}_{1}, \ldots, \mcl{X}_{n-1}, \mcl{Y}_{1}, \ldots, \mcl{Y}_{n-1}} \}\ .
\end{align}
\end{subequations}

\begin{itemize}
    \item When $\mcl{K} = \strat{\mcl{X}_{1}, \ldots, \mcl{X}_n, \mcl{Y}_{1}, \ldots, \mcl{Y}_n}$ a simple inductive calculation shows that we can bound the trace with $A = \dim(\mcl{X}_1 \otimes \cdots \otimes \mcl{X}_n) = \prod_{i=1}^n \dim(\mcl{X}_i)$.

    \item When $\mcl{E} = \strat{\mcl{X}_{1}, \ldots, \mcl{X}_n, \mcl{Y}_{1}, \ldots, \mcl{Y}_n}$, we have $C \leq A = \prod_{i=1}^n \dim(\mcl{X}_i)$ using the same argument that we used for quantum channels.
    We note that this $C$ may not be tight as it was for channels, but, regardless, it works to give us a meaningful regret bound.
\end{itemize}

The Choi representations of co-strategies can be defined in a similar way, below
\begin{subequations}
\begin{align}
    \costrat{\mcl{X}_{1}, \ldots, \mcl{X}_n, \mcl{Y}_{1}, \ldots, \mcl{Y}_n}
    = \{ & R \in \Pos{\mcl{Y}_{1} \otimes \cdots \otimes \mcl{Y}_{n} \otimes \mcl{X}_{1} \otimes \cdots \otimes \mcl{X}_n} : \\
        & R = R' \otimes \identity_{\mcl{Y}_n}, \; \text{ for some }\\
        &  R' \in \Pos{\mcl{Y}_{1} \otimes \cdots \otimes \mcl{Y}_{n-1} \otimes \mcl{X}_{1} \otimes \cdots \otimes \mcl{X}_n}, \\
        & \ptrace{\mcl{X}_n}{R'} \in \costrat{\mcl{X}_{1}, \ldots, \mcl{X}_{n-1}, \mcl{Y}_{1}, \ldots, \mcl{Y}_{n-1}} \}\ .
\end{align}
\end{subequations}

As it stands now, for any strategy $Q$ and co-strategy $R$, we have $\inner{Q}{R} = 1$, which is not interesting for online learning.
What we want is a \emph{measuring} co-strategy, where a measurement is made at the end of the co-strategy.
If we have a matrix $X \geq 0$ such that $X \leq R$ for some co-strategy $R$, then the probability of seeing $X$ when strategy $Q$ interacts with $R$, then measured, is given by $\inner{Q}{X}$.
This is the set of co-objects we seek in this application, and is defined below.
\begin{subequations}
    \begin{align}
        \downarrow\!\costrat{\mcl{X}_{1}, \ldots, \mcl{X}_n, \mcl{Y}_{1}, \ldots, \mcl{Y}_n} = \{ & X \geq 0 : X \leq R \text{ for some } \\
            & R \in \costrat{\mcl{X}_{1}, \ldots, \mcl{X}_n, \mcl{Y}_{1}, \ldots, \mcl{Y}_n} \}\ .
    \end{align}
\end{subequations}

\begin{itemize}
    \item When $\mcl{K} = \downarrow\!\costrat{\mcl{X}_{1}, \ldots, \mcl{X}_n, \mcl{Y}_{1}, \ldots, \mcl{Y}_n}$, a similar calculation as in the case of strategies shows that we can set $A = \prod_{i=1}^n \dim(\mcl{Y}_i)$.
    \item When $\mcl{E} = \downarrow\!\costrat{\mcl{X}_{1}, \ldots, \mcl{X}_n, \mcl{Y}_{1}, \ldots, \mcl{Y}_n}$, we have $C \leq A = \prod_{i=1}^n \dim(\mcl{Y}_i)$.
    As in the case of strategies, this bound for $C$ may not be tight.
\end{itemize}

\begin{corollary}[Online learning of quantum strategies]
\label{app:strats}
For $\mcl{K} = \strat{\mcl{X}_{1}, \ldots, \mcl{X}_n, \mcl{Y}_{1}, \ldots, \mcl{Y}_n} $ and also $\mcl{E} = \downarrow\!\costrat{\mcl{X}_{1}, \ldots, \mcl{X}_n, \mcl{Y}_{1}, \ldots, \mcl{Y}_n}$, we have
    \begin{equation}
        \regret{T} \leq \left( 4 B
            \left( \prod_{i=1}^n \dim(\mcl{Y}_i) \right)
            \left( \prod_{i=1}^n \dim(\mcl{X}_i) \right)
            \sqrt{ \sum_{i=1}^n \ln(\dim(\mcl{X}_i)) + \sum_{i=1}^n \ln(\dim(\mcl{Y}_i))}
        \right) \sqrt{T}\ .
    \end{equation}
\end{corollary}

Flipping the respective roles in the result above, we get the following.

\begin{corollary}[Online learning of measuring co-strategies]\label{app:costrats}
    \newcommand{\npt}{\hspace*{-.92pt}}
    Let $\mcl{K}\npt= \downarrow\!\costrat{\mcl{X}_{1}\npt, \ldots\npt, \mcl{X}_n\npt, \mcl{Y}_{1}\npt, \ldots\npt, \mcl{Y}_n}$ and $\mcl{E} = \strat{\mcl{X}_{1}, \ldots, \mcl{X}_n, \mcl{Y}_{1}, \ldots, \mcl{Y}_n}$.
    When it is the case\footnote{The case when $\prod_{i=1}^n \dim(\mcl{X}_i) = 2$ is interesting as well as it corresponds to a single qubit message, possibly in the middle of the interaction.
    In this case, the regret bound can be worked out as well, it is just unwieldy to write down.} that $\prod_{i=1}^n \dim(\mcl{X}_i) \geq 3$,
    we have
    \begin{equation}
        \regret{T} \leq \left( 4 B
            \left( \prod_{i=1}^n \dim(\mcl{Y}_i) \right)
            \left( \prod_{i=1}^n \dim(\mcl{X}_i) \right)
            \sqrt{ \sum_{i=1}^n \ln(\dim(\mcl{X}_i)) + \sum_{i=1}^n \ln(\dim(\mcl{Y}_i))}
        \right) \sqrt{T}\ .
    \end{equation}
\end{corollary}

Notice that the regret bounds in \cref{app:strats,app:costrats} are the same, which may not be too surprising given the almost symmetric nature of strategies and co-strategies.

We conclude by remarking that there are many other interesting sets of quantum objects that our bound can be applied to in order to get sublinear regret.
We refer the interested reader to the excellent book by \cite{watrous2018theory} for further discussion on the quantum concepts discussed in this section as well as other possible examples.

\section{Minimizers are positive definite}\label{sec:pd}

\begin{lemma}\label{lemma:minimzerIsPD}
    Suppose $\alpha \identity_\mcl{X} \in \mcl{K}$ for some $\alpha > 0$.
    Any minimizer of $\F{E}{X}$ over $\mcl{K}$ is positive definite.
\end{lemma}

\begin{proof}
Since $\mcl{K}$ is compact and $\F{E}{}$ is continuous, we know that there exists a minimizer and, moreover, since $\F{E}{}$ is strictly convex, we know that this minimizer is unique.
Denote this minimizer as $W \in \mcl{K}$ and, for brevity, $A = \alpha \identity_\mcl{X}$.
Then we have
\begin{equation} \label{eq:pd-min}
    \F{E}{W} < \F{E}{(1-t)W + tA}, \; \forall t \in (0,1)
\end{equation}
since $(1-t)W + tA \in \mcl{K}$ by convexity.
Rearranging \cref{eq:pd-min}, we have
\begin{equation}
    \inner{E}{W - A} < \frac{\entr{W} - \entr{(1-t)W + tA}}{t} = \frac{\entr{W} - \entr{W + t(A - W)}}{t}.
\end{equation}
Since the left-hand side is a fixed constant, all we need to show is that the right-hand side can be made arbitrarily small as $t$ decreases to $0$ when $W$ is rank-deficient to get a contradiction.

The right-hand side is the directional derivative of $\reg{}$ at $W$ in the direction $A - W$.
Since $A = \alpha \identity$, it commutes with $W$ which simplifies the expression that we wish to work with.
Let $\lambda_1, \ldots, \lambda_r > 0$ be the positive eigenvalues of $W$ noting that $r < n \coloneqq \dim(\mcl{X})$ since $W$ is rank-deficient.
Define the function $f(x) \coloneqq x \ln x$.
We can now write
\begin{equation}
    \dfrac{\entr{W} - \entr{W + t(A-W)}}{t}
    = \sum_{i=1}^r \frac{f(\lambda_i + t (\alpha - \lambda_i)) - f(\lambda_i)}{t}
      + (n - r) \frac{f(t \alpha)}{t}.
\end{equation}
For $x > 0$, the derivative is given by $f'(x) = 1 + \ln(x)$, and thus
\begin{equation}
    \frac{f(\lambda_i + t (\alpha - \lambda_i)) - f(\lambda_i)}{t} \to \left( 1 + \ln(\lambda_i) \right) \left( \alpha - \lambda_i \right) \quad \text{ as } \quad t \to 0^+\ ,
\end{equation}
which is finite.
However,
\begin{equation}
    \frac{f(t \alpha)}{t} = \alpha \ln(t \alpha) \to - \infty \quad \text{ as } \quad t \to 0^+\ .
\end{equation}
Thus, the right-hand side of \cref{eq:pd-min} can be made arbitrarily small.
\end{proof}

\section{Bounding the diameter}\label{sec:diameter}

The following lemma is helpful in proving \cref{lem:dbound}.

\begin{lemma}\label{lem:non-normalized-entropy-bounds}
For $\alpha > 0$, we have
\begin{equation}\label{eq:non-norm-max-ent}
    \max \{ S(X) : \trace{X} \leq \alpha, X \in \Pos{\mcl{X}} \} = \begin{cases} e^{-1}\dim(\mcl{X}) & \alpha \geq e^{-1}\dim(\mcl{X}) \\
    -\alpha \ln(\alpha) + \alpha\ln(\dim(\mcl{X})) & \text{otherwise} \ ,
    \end{cases}
\end{equation}
which is attained at $e^{-1}\identity_{\mcl{X}}$ in the first case and $\dfrac{\alpha}{\dim(\mcl{X})}\identity(\mcl{X})$ otherwise.
Moreover, we have
\begin{equation}
    \min \{ \entr{X} : \trace{X} \leq \alpha , X \in \Pos{\mcl{X}} \} = \begin{cases}
    0 & \alpha \leq 1 \\
    -\alpha \ln(\alpha) & \text{otherwise} \ ,
    \end{cases}
\end{equation}
which in the first case is attained by the zero matrix or any any rank one positive semidefinite matrix $vv\adj$ where $\| v \|_2 = 1$, or in the second case at any rank one positive semidefinite matrix $\alpha \cdot vv\adj$ where $\| v \|_2 = 1$.
\end{lemma}

\begin{proof}
    We consider any $X \in \Pos{\mcl{X}}$ with $\trace{X} = \alpha > 0$.
    Let $\Pi_{\supp{X}}$ be the projector onto the support of $X$.
    Then,
    \begin{subequations}
    \begin{align}
        \entr{X}
        =& -\trace{X \ln(X)} \\
        =& -\trace{X} \trace*{ \widehat{X} \ln\delarg[\big]{ \trace{X} \widehat{X} } } \\
        =& -\alpha \trace*{ \widehat{X} \Big\{ \ln\delarg[\big]{\trace{X}} \Pi_{\supp{X}} + \ln(\widehat{X}) \Big\} } \\
        =& -\alpha \left( \ln(\alpha) \trace{\widehat{X} \, \Pi_{\supp{X}}} + \trace{\widehat{X} \ln(\widehat{X})} \right) \\
        =& -\alpha \ln(\alpha) + \alpha \entr{}(\widehat{X}) \ ,
    \end{align}
    \end{subequations}
    where the second equality uses the definition $\hat{X} \coloneqq \trace{X}^{-1}X$, the third may be determined using the spectral decomposition and definition of $\alpha$, the fourth uses the definition of $\alpha$ and linearity of trace, and the last is again using the definition of entropy.
    Thus, we just need to maximize or minimize the entropy over the set of density matrices and optimize over $\alpha$.

    We begin with maximizing.
    In this case, as stated in \cref{eq:entropy-bounds}, it is well known that
    \begin{equation}
        \max \{ \entr{X} : \trace{X} = 1, X \geq 0 \} = \ln(\dim(X)) \ ,
    \end{equation}
    and is attained at $\dfrac{1}{\dim(\mcl{X})} \identity_{\mcl{X}}$.
    Thus, in this case we are interested in maximizing
    \begin{equation}
        f(x) \coloneqq -x\ln(x) + x\ln(\dim(\mcl{X}))\ .
    \end{equation}
    Now note that $f(x)$ takes the value zero at the points $\{0,\dim(\mcl{X})\}$ and is positive between these points.
    Thus, we want the maximum.
    Taking the derivative and setting it equal to zero,
    \begin{equation}
        -1 + \ln(\dim(\mcl{X})/x) = 0 \implies x = e^{-1}\dim(\mcl{X}).
    \end{equation}
    Simplifying $f(e^{-1}\dim(\mcl{X}))$ obtains the first case of \cref{eq:non-norm-max-ent}, which is then obtained at
    \begin{equation}
        e^{-1}\dim(\mcl{X}) \frac{1}{\dim(\mcl{X})}\identity_{\mcl{X}} = e^{-1}\identity_{\mcl{X}}.
    \end{equation}
    If $\alpha < e^{-1}\dim(\mcl{X})$, then $g(x)$ is monotonically increasing over the interval $[0,\alpha]$, so the optimal value is attained by the largest value possible, $\alpha$ and is attained by the maximally mixed state scaled by that value.

    Similarly, as stated in \cref{eq:entropy-bounds}, we have
    \begin{equation}
        \min \{ \entr{X} : \trace{X} = 1, X \geq 0 \} = 0
    \end{equation}
    attained at any pure state $vv\adj$ (so $\| v \|_2 = 1$).
    Thus, we are interested in minimizing the function $g(x) \coloneqq -x\ln(x)$.
    Note $g(x)$ is zero at $x \in \{0,1\}$, positive over the interval $[0,1]$, and otherwise is negative.
\end{proof}

We are now ready to prove \cref{lem:dbound}.

\begin{proof}[Proof of \cref{lem:dbound}]
    Note that the diameter can only increase if we relax the set.
    Thus, we can bound the diameter over the sets considered in \cref{lem:non-normalized-entropy-bounds}.
    If $\mcl{K}$ has variable traces, then the diameter bound follows immediately from \cref{lem:non-normalized-entropy-bounds}.
    If $\mcl{K}$ has constant trace, say $A \geq 1$, then our diameter bound follows immediately from
    \begin{equation}
        S(X) = -A \ln(A) + A S(\hat{X})\ ,
    \end{equation}
    where $\hat{X} = \dfrac{1}{A} X \in \density{\mcl{X}}$ and the fact that $S(\hat{X}) \in [0, \ln(\dim(\mcl{X}))]$.
\end{proof}

\section{Generalized Pinsker's inequality}\label{sec:pinsker}

In this section, we establish our generalization of Pinsker's inequality (\cref{thm:generalized-Pinsker's-inequality}), which is \cref{cor:quantum-gen-pinsker} in this section.
We use this as an intermediary result to ultimately bound the regret.
Namely, we use it as a lemma in establishing \cref{lemma:InnerProductBound}.
This result was not needed in \cite{aaronson2018online} because the authors only considered the set of quantum states where the trace cannot vary.
As we do not make this guarantee in our setting, we must establish this generalization.

For completeness, we explain how our proof method differs from the standard method for establishing Pinsker's inequality, $D(P||Q) \geq \frac{1}{2}\|P-Q\|_{1}^{2}$.
By properties of the trace norm $\|\cdot \|_{1}$ for $X \in \Herm{\mcl{X}}$ and the data processing inequality of relative entropy (\cref{lemma:rel-ent-DPI}), it ultimately suffices to establish the result classically.
Again by the data processing inequality, it suffices to establish Pinsker's inequality for Bernoulli distributions.
While there is a canonical method of establishing the Bernoulli distribution case (see \cite{Wilde-Book} for example), a particularly simple method provided in \cite{polyanskiy-2023a} is using the remainder form of Taylor's theorem.
However, the proof that makes use of the remainder form of Taylor's theorem relies on $\|\mbf{p} - \mbf{q}\|_{1}^{2} = 2(p-q)^{2}$ for Bernoulli distributions defined by $p,q \in [0,1]$.
This relies on the vectors both summing to one, which we cannot guarantee in our setting as we vary the trace.
To resolve this, we extend the proof method by using Taylor's theorem for multivariate vectors.
This ultimately allows us to establish our result.
We remark this proof method has applications beyond those relevant for this work, which we investigate in a separate paper \citep{Gen-Pinsker-Ineqs}.

The following proofs rely on the following well-known results.
The first is data processing of relative entropy, a proof of the form we use may be found in \cite{Wilde-Book}.
\begin{lemma}[Data processing]\label{lemma:rel-ent-DPI}
    Let $P,Q \in \Pos{\mcl{X}}$.
    Let $\mcl{N} \in \channel{\mcl{X}}{\mcl{Y}}$.
    Then,
    \begin{equation}
        D(\mcl{N}(P)\Vert\, \mcl{N}(Q)) \leq D(P\Vert Q) \ .
    \end{equation}
\end{lemma}

The second needed result is a version of Multivariate Taylor's theorem for twice-differentiable functions.
\begin{lemma}\label{lem:first-order-taylor}
     Let $f: \reals^{n} \to \reals$ be $C^{2}$ on an open convex set $S$.
     If $\mbf{a} \in S$ and $\mbf{a} + \mbf{h} \in S$ then
     \begin{equation}
         f(\mbf{a}+\mbf{h}) = f(\mbf{a}) + \langle \nabla f (\mbf{a}), \mbf{h}\rangle + \int^{1}_{0} (1-t)\ \mbf{h}^{T} H_{f}\vert_{\mbf{a}+t\mbf{h}} \mbf{h}\ dt \ .
     \end{equation}
\end{lemma}
This is a direct calculation from multivariate Taylor's theorem in terms of an integral remainder, which we provide in \cref{appx:taylor} for completeness.

We now begin to establish our result.
We first define the following function that is well-defined for any strictly positive vector $\mbf{q} > \mbf{0}$:
\begin{align}
    f_{\mbf{q}}(\mbf{p}) \coloneqq \sum_{i \in [n]} p_{i}\ln(p_{i}/q_{i}) = D(\mbf{p}||\mbf{q}) \ ,
\end{align}
which is defined on the open convex set $S = (0,\infty)^{\times n}$.

\begin{theorem}[Classical version]\label{thm:classical-version}
    Consider $n$-dimensional vectors $\mbf{q}, \mbf{p} \geq \mbf{0}$ and define $\check{\mbf{q}}$ as $\mbf{q}$ restricted to the support of $\mbf{p}$ and $\hat{\mbf{q}} = \mbf{q} - \check{\mbf{q}}$
    \begin{align}
        D(\mbf{p}||\mbf{q}) - \langle \mbf{p} \rangle + \langle \check{\mbf{q}} \rangle + \|\hat{\mbf{q}}\|_{2}^{2}
        \geq  \frac{C}{2}\|\mbf{p}-\mbf{q}\|_{2}^{2} \ , \label{eq:2-norm-pinsker}
    \end{align}
    where $C^{-1} \coloneqq \|\mbf{p}\oplus \mbf{q}\|_{\infty}$.
    This further implies
    \begin{align}
        D(\mbf{p}||\mbf{q}) - \langle \mbf{p} \rangle + \langle \check{\mbf{q}} \rangle + \|\hat{\mbf{q}}\|_{2}^{2}
        \geq  \frac{C_{2}}{4}\|\mbf{p}-\mbf{q}\|_{1}^{2} \ , \label{eq:1-norm-pinsker}
    \end{align}
    where $C_{2}^{-1} \coloneqq \|\mcl{A} \mbf{p} \oplus \mcl{A} \mbf{q}\|_{\infty}$ and $\mcl{A} \coloneqq \sum_{i \in I} e_0 e_i\adj + \sum_{i \not \in I} e_1 e_i\adj$,
    where $e_i$ is the $i$th standard basis vector (counting from $0$).
    Moreover, we note that on the right-hand side of these inequalities $C, C_{2}$ can be replaced by $\widetilde{C}^{-1} \coloneqq \max \{ \langle \mbf{p} \rangle, \langle \mbf{q} \rangle \}$.
    Furthermore, in the special case that $\mbf{q},\mbf{p} \in (0,\infty)^{\times 2}$, $C_{2}$ may be replaced with $C$.
\end{theorem}
\begin{proof}
    For the bulk of the proof, we assume $\mbf{p},\mbf{q} > \mbf{0}$, which means the $\|\hat{\mbf{q}}\|_{2}^{2}$ term is zero and not relevant.
    We lift this at the end of the proof.
    The partial derivatives of $f_{\mbf{q}}(\mbf{p})$ are
    \begin{align}
        \frac{\partial f_{\mbf{q}}}{\partial p_{i}} = 1 + \ln(p_{i}/q_{i})\ , \qquad \frac{\partial^{2} f_{\mbf{q}}}{\partial p_{i} \partial p_{j}} = \delta_{i,j}p_{i}^{-1}\ , \qquad  & i,j \in [n] \ .
    \end{align}
    This implies
    \begin{align}
        H_{f_{\mbf{q}}} = \sum_{i \in [n]} p_{i}^{-1} E_{i,i}\ ,\qquad
        \nabla f_{\mbf{q}} = \sum_{i \in [n]} (
            1 + \ln(p_{i}/q_{i})) \mbf{e}_{i}\ ,
    \end{align}
    where $\mbf{e}_{i}$ indicates the standard basis notation.
    Now we are going to choose $\mbf{a} = \mbf{q}$, $\mbf{p} = \mbf{a} + \mbf{h}$, which means $\mbf{h} = \mbf{p}-\mbf{q}$.
    Applying Lemma \ref{lem:first-order-taylor},
    \begin{subequations}
    \begin{align}
        f_{\mbf{q}}(\mbf{p}) =& f_{\mbf{q}}(\mbf{q}) + \langle \nabla f_{\mbf{q}}(\mbf{q}), \mbf{h} \rangle + \int^{1}_{0} (1-t)\ \mbf{h}^{T} H_{f_{\mbf{q}}}\vert_{\mbf{a}+t\mbf{h}} \mbf{h}\ dt \\
        =& \langle \mbf{p}-\mbf{q} \rangle + \int^{1}_{0} (1-t)\ \mbf{h}^{T} H_{f}\vert_{\mbf{a}+t\mbf{h}} \mbf{h}\ dt \ ,
    \end{align}
    \end{subequations}
    where we have used $f_{\mbf{q}}(\mbf{q}) = 0$ and $\nabla f_{\mbf{q}}(\mbf{q}) = \sum_{i} \mbf{e}_{i}$.
    Pushing the $\langle \mbf{p} - \mbf{q} \rangle$ term to the left-hand side of the equality gets us what we want to lower bound.
    Since $\mbf{a}+t\mbf{h} = \mbf{q} + t(\mbf{p}-\mbf{q}) = (1-t)\mbf{q} + t\mbf{p}$,
    \begin{align}
        f_{\mbf{q}}(\mbf{p}) - \langle \mbf{p} \rangle + \langle \mbf{q} \rangle =& \int^{1}_{0} (1-t)\left[ \sum_{i} \frac{(p_{i}-q_{i})^{2}}{(1-t)q_{i} +tp_{i}} \right] dt \ .
    \end{align}
    At this point it suffices to lower-bound the right-hand side.
    We do this in the following fashion,
    \begin{align}
        \int^{1}_{0} (1-t)\left[ \sum_{i} \frac{(p_{i}-q_{i})^{2}}{(1-t)q_{i} +tp_{i}} \right] dt
        \geq C \int^{1}_{0} (1-t)\left[ \sum_{i} (p_{i}-q_{i})^{2} \right] dt
        = \frac{C}{2} \|\mbf{p}-\mbf{q}\|^{2}_{2} \ ,
    \end{align}
    where one lets $C$ be any constant that lower bounds $[(1-t)q_{i}+tp_{i}]^{-1}$ for all $i$.
    An easy choice is $C = \|\mbf{p}\oplus \mbf{q}\|_{\infty}^{-1}$.
    This proves \cref{eq:2-norm-pinsker}.
    A looser choice that may be convenient is $\max\{\langle \mbf{p} \rangle, \langle \mbf{q} \rangle \}^{-1}$, which proves our moreover statement.

    We now convert this to a $1$-norm bound.
    First, $\frac{1}{\sqrt{|\text{supp}(\mbf{x})|}}\|\mbf{x}\|_{1} \leq\|\mbf{x}\|_{2}$, so for $2$-dimensional vectors, we have $C/4 \|\mbf{p}-\mbf{q}\|_{1}^{2} \leq C/2 \|\mbf{p}-\mbf{q}\|_{2}^{2}$.
    This explains why we can use $C$ in \cref{eq:1-norm-pinsker} when $\mbf{p},\mbf{q}$ are two-dimensional.

    We now use data processing (\cref{lemma:rel-ent-DPI}) to remove scaling in the dimension.
    Define the set $I \coloneqq \{i: p_{i} - q_{i} > 0\}$ and the classical channel in matrix form $\mcl{A} \coloneqq \sum_{i \in I} e_0 e_i\adj + \sum_{i \not \in I} e_1 e_i\adj$.
    It follows $\|\mbf{p}-\mbf{q}\|_{1} = \|\mcl{A}(\mbf{p}) - \mcl{A}(\mbf{q})\|_{1}$.
    Using that $f_{\mbf{q}}(\mbf{p}) = D(\mbf{p}||\mbf{q})$ and that $\mcl{A}$ is CPTP,
    \begin{subequations}
    \begin{align}
        f_{\mbf{q}}(\mbf{p}) - \langle \mbf{p} \rangle + \langle \mbf{q} \rangle =& D(\mbf{p}||\mbf{q}) - \langle \mbf{p} \rangle + \langle \mbf{q} \rangle \\
        \geq& D(\mcl{A} \mbf{p}||\mcl{A} \mbf{q}) - \langle \mcl{A}\mbf{p} \rangle + \langle \mcl{A}\mbf{q} \rangle \\
        \geq& \frac{C_{2}}{4} \|\mcl{A} \mbf{p} - \mcl{A} \mbf{q}\|^{2}_{1} \\
        =& \frac{C_{2}}{4} \|\mbf{p} - \mbf{q}\|_{1}^{2} \ ,
    \end{align}
    \end{subequations}
    where $C_{2}$ is just $C$ under this coarse-graining $\mcl{A}$.
    This proves \cref{eq:1-norm-pinsker}.
    Finally, we see can replace $C_{2}$ by $\widetilde{C}$ because $\mcl{A}$ is CPTP.
    This completes everything for strictly positive vectors.

    To lift this to non-negative vectors, note that if the support of $\mbf{p} \not \ll \mbf{q}$, then $D(\mbf{p}||\mbf{q})$ is infinite and the result is trivial.
    Thus, we assume the support of $\mbf{p} \ll \mbf{q}$.
    Then $D(\mbf{p}||\mbf{q}) = D(\mbf{p}||\check{\mbf{q}})$, so we apply Theorem \ref{thm:classical-version} to this.
    Then note that $\|\mbf{p} - \mbf{q}\|_{2}^{2} = \|\mbf{p}-\check{\mbf{q}}\|_{2}^{2} + \|\hat{\mbf{q}}\|_{2}^{2}$, so we just add it to both sides (properly scaled) to complete the proof.
\end{proof}

Before proceeding to the quantum extension, we note if you let $\mbf{p},\mbf{q}$ be normalized distributions and use $\widetilde{C} = 1$, the above gives $D(\mbf{p}||\mbf{q}) \geq \Delta(\mbf{p},\mbf{q})^{2}$, which is strictly more loose than Pinsker's inequality unless $\mbf{p} = \mbf{q}$.
However, this does not change the fact this result is sharp in some situations as we present.

\begin{proposition}[Sharpness]\label{prop:sharpness}
    For two-dimensional probability distributions, the inequality \cref{eq:2-norm-pinsker} and the lower bound $\frac{C}{4}\|\mbf{p}-\mbf{q}\|_{1}^{2}$ are sharp.
\end{proposition}
\begin{proof}
    Since the right-hand side of \cref{eq:2-norm-pinsker} upper bounds $\frac{C}{4}\|\mbf{p}-\mbf{q}\|_{1}^{2}$, it suffices to focus on the latter.
    Now note that if we have normalized probability distributions, we can rewrite the inequality as $D(\mbf{p}||\mbf{q}) \geq C \, \Delta(\mbf{p},\mbf{q})^{2}$ where $\Delta(\mbf{p},\mbf{q})^{2} \coloneqq \frac{1}{2}\|\mbf{p}-\mbf{q}\|_{1}$ is the trace distance (total variational distance).
    Now consider the sequence of distributions $\mbf{p}_{n} = \begin{bmatrix} 1/2 + 1/n & 1/2 - 1/n \end{bmatrix}\T$ and $\mbf{q}_{n} = \begin{bmatrix} 1/2 & 1/2 \end{bmatrix}\T$.
    Then for every $n$, the coefficent $C$, is $C_{n} = (\frac{1}{2}+\frac{1}{n})^{-1}$.
    We now consider the limit of the ratio
    \begin{subequations}
    \begin{align}
        \lim_{n \to \infty} \frac{D(\mbf{p}_{n}||\mbf{q}_{n})}{C_{n} \,  \Delta(\mbf{p}_{n},\mbf{q}_{n})^{2}} =& \lim_{n \to \infty} \frac{1/2 \chi^{2}(\mbf{p}_{n}||\mbf{q}_{n})}{C_{n} (1/n^{2})} \\
        =& \lim_{n \to \infty} \frac{(1/2)(4/n^{2})}{(1/2+1/n)^{-1}(1/n^{2})} \\
        =& \lim_{n \to \infty} 1 + 2/n \\
        =& 1 \ ,
    \end{align}
    \end{subequations}
    where first equality is using $\Delta(\mbf{p}_{n},\mbf{q}_{n}) = 1/n$ and that relative entropy is locally $\chi^{2}$-like \citep{polyanskiy-2023a} (but we are working in base $e$), the second is using $\chi^{2}(\mbf{p}||\mbf{q}) = \sum_{i \in [n]} q_{i}(p_{i}/q_{i}-1)^{2}$, the third is simplifying, and the last is taking the limit.
    This completes the proof.
\end{proof}

Finally, we extend our classical result to the quantum setting.
\begin{corollary}[Quantum Version]\label{cor:quantum-gen-pinsker}
    For any positive semidefinite operators $P,Q \geq 0$,
    \begin{align}
        \relentr{P}{Q} - \trace{P} + \trace{Q}
        \geq  \frac{C_{Q}}{4}\|P-Q\|_{1}^{2} \geq \frac{\widetilde{C}_{Q}}{4}\|P-Q\|_{1}^{2} \ , \label{eq:q-pinsker}
    \end{align}
    where $C_{Q}^{-1} \coloneqq \max \{\trace{\Pi P}, \trace{\Pi Q}, \trace{P} - \trace{\Pi P}, \trace{Q} - \trace{\Pi Q}\}$ where $\Pi$ is the projector onto the positive eigenspace of $P-Q$ and  $\widetilde{C}_{Q}^{-1} \coloneqq \max\{\trace{P},\trace{Q}\}$.
    Moreover, the first inequality is sharp.
\end{corollary}
\begin{proof}
    First note for any $P,Q \geq 0$, there exists a binary measurement channel $\mcl{E}$ such that we have $\|P-Q\|_{1} = \|\mcl{E}(P)-\mcl{E}(Q)\|_{1}$.
    This is achieved by defining the projector onto the positive eigenspace of $P-Q$, $\Pi$, and defining $\mcl{E}(X) = \trace{\Pi X} E_{0,0} + \trace{(\identity-\Pi)X} E_{1,1}$ where we define $E_{0,0} = \begin{bmatrix} 1 & 0 \\ 0 & 0 \end{bmatrix}$ and $E_{1,1} = \begin{bmatrix} 0 & 0 \\ 0 & 1 \end{bmatrix}$.
    This is because, by linearity, we then obtain
    \begin{equation}
        \|\mcl{E}(P)-\mcl{E}(Q)\|_{1} = |\trace{(P-Q)_{+}}| + |\trace{(P-Q)_{-}}| = \|P-Q\|_{1} \ .
    \end{equation}
    Then we have
    \begin{subequations}
    \begin{align}
        \relentr{P}{Q} - \trace{P} + \trace{Q} \geq& D(\mcl{E}(P)||\mcl{E}(Q)) - \trace{\mcl{E}(P)} + \trace{\mcl{E}(Q)} \\
        \geq & \frac{C_{Q}}{4}\|\mcl{E}(P) - \mcl{E}(Q)\|_{1}^{2} \\
        =& \frac{C_{Q}}{4}\|P-Q\|_{1}^{2} \\
        \geq & \frac{\widetilde{C}_{Q}}{4} \|P-Q\|_{1}^{2} \ ,
    \end{align}
    \end{subequations}
    where the first inequality is DPI and that $\mcl{E}$ is CPTP, the second inequality is Theorem \ref{thm:classical-version} where we use $L_{1}$-norm version with the $C$ coefficient for binary distributions, the third is our choice of $\mcl{E}$, and the fourth is our definition of $C_{Q}, \widetilde{C}_{Q}$.
    This establishes the inequalities so long as $\mcl{E}(P)$ and $\mcl{E}(Q)$ are full rank.
    If they are not, this means $\mcl{E}(P) = E_{0,0}$ and $\mcl{E}(Q) = E_{1,1}$ (up to labeling) and thus the inequality holds trivially as $D(P||Q) = +\infty$.

    Finally, to see that this is tight, note that if $P,Q$ are just the matrix versions of the Bernoulli distributions considered in Proposition \ref{prop:sharpness}, then the measurement channel does not do anything, so we inherit sharpness from Proposition \ref{prop:sharpness}.
    This completes the proof.
\end{proof}

\section{Derivative of the function \texorpdfstring{$\mbf{\Phi_{E}}$}{Phi}} \label{sec:deriv-of-func-phi-E}

In this section we establish that the function $\Phi_{E}(X) \coloneqq \langle E , X \rangle + \trace{X \ln(X)}$ has a gradient for all $X \in \Pd{\mcl{X}}$ and that it is given by
\begin{align}
    \nabla \Phi_{E}(X) = E + \identity + \ln(X) \ .
\end{align}

This claim was stated as \cref{lem:OpDerivative}, which we prove here in a somewhat roundabout manner.
Specifically, we use that for a function $f:\Herm{\mcl{X}} \to \reals$, the relationship between the Fr\'{e}chet derivative, the directional derivative, and its relation to the gradient allows us to ``extract" the gradient of $f$ whenever it exists.
We summarize this relationship as a proposition after a definition.
These facts are known and have been used in quantum information theory previously, e.g. \cite{girard_convex_2014,Coutts-2021a}.
We restate them here for clarity.

\begin{definition}[\cite{Bhatia-1997}]\label{def:Frechet-deriv} Given a map $f:\mcl{A} \to \Herm{\mcl{Y}}$ where $\mcl{A} \subseteq \Herm{\mcl{X}}$, the map is \textit{Fr\'{e}chet} differentiable at $X \in \Herm{\mcl{X}}$ if there exists a linear map $\Phi:\Herm{\mcl{X}}\to \Herm{\mcl{Y}}$ such that it holds
\begin{align}\label{eq:frechet-deriv-defn}
    \lim_{Z\to 0} \frac{\|f(X+Z) - f(X) - \Phi(Z)\|}{\|Z\|} = 0 \ .
\end{align}
\end{definition}
When such a map $\Phi$ exists, it is unique and we denoted it as $Df(X)$.
This derivative has many known properties, such as linearity.
\begin{proposition}\label{prop:frech-relation}(See \cite{Bhatia-1997,Coutts-2021a})
    When $f$ is differentiable at $X$, the map is the directional derivative of $f$ at $X$ for any $Z \in \Herm{\mcl{X}}$.
    That is, when $f$ is differentiable at $X$,
    \begin{align}\label{eq:frech-deriv-to-direc-deriv}
        Df(X)[Z] = \frac{d}{dt} f(X+tZ) \vert_{t=0}
    \end{align}
    for all $Z \in \Herm{\mcl{X}}$.
    Moreover, for $f:\mcl{A} \to \reals$, when $f$ is differentiable at $X$,
    \begin{align}\label{eq:frech-deriv-to-grad}
        Df(X)[Z] = \langle \nabla f(X), Z \rangle
    \end{align}
    for all $Z \in \Herm{\mcl{X}}$.
\end{proposition}

Combining the two points of the previous proposition, whenever we can prove $f$ is differentiable at $X$, if we can determine its directional derivative, we have also determined the gradient of $f$ at $X$, which is what we want.
To do this, we need the following lemma, which is a strengthening of a known result \cite{Carlen-2010}.
The basic idea is to prove the result for polynomials and then extend the result via continuity.
We remark our result is in fact a generalization of what is given in \cite{Carlen-2010} as we do not require $H$ to be positive definite in the following statement.
\begin{lemma}\label{lem:Carlen-Generalization}
    Let $I \subset \reals$ be an interval.
    Let $f : I \to \reals$ be once continuously differentiable, $f \in C^{1}(I)$.
    Let $X \in \Herm{\mcl{X}}$ such that its eigenvalues are contained within the interval $I$, $\spec{X} \subset I$.
    Then, the extension of the map $f$ to Hermitian matrices via the spectral decomposition is Fr\'{e}chet differentiable at $X$ and we have for all $H \in \Herm{\mcl{X}}$,
    \begin{align}\label{lem:thm-version-of-carlen-ext}
         D\trace{f(X)}[H] = \frac{d}{dt} \trace{f(X+tH)} \Big \vert_{t=0} = \trace{f'(X)H} \ .
    \end{align}
\end{lemma}
\begin{proof}
    See \cref{appx:proofCarlen}.
\end{proof}

We are now ready to establish the gradient of $\Phi_E$ by applying the above result.
\begin{proof}[Proof of \cref{lem:OpDerivative}]
    First note that $\trace{X \ln (X)} = \trace{f(X)}$ where  $f(t)\coloneqq t\ln(t)$, i.e., it is a well-known trace function.
    By linearity of the Fr\'{e}chet derivative, whenever it exists for $H \in \Herm{\mcl{X}}$,
    \begin{equation}
        D\Phi_{E}(X)[H] = D\inner{E}{X}[H] + D\trace{f(X)}[H] \ .
    \end{equation}
    Moreover, recall that if two functions are Fr\'{e}chet differentiable at a point, then their sum is also Fr\'{e}chet differentiable at that point as may be verified rather directly from Definition \ref{def:Frechet-deriv} and the uniqueness of the Fr\'{e}chet derivative.
    From these two points, it suffices to consider the differentiability of the two functions independently.

    First, for all $X \in \Herm{\mcl{X}}$, $D\inner{E}{X}[H] = \inner{E}{H}$.
    This may be verified by substituting this map as $\Phi$ in \cref{eq:frechet-deriv-defn} and verifying it satisfies the condition.
    Note this means that $\inner{E}{X}$ is differentiable on $X \in \Herm{\mcl{X}}$.

    Next, using $\frac{d}{dt} t \ln(t) = 1 + \ln(t)$ and Lemma \ref{lem:Carlen-Generalization}, which apply so long as $\spec{X} \in (0,+\infty)$,
    \begin{align}
        D\trace{f(X)}[H] = \inner{\identity + \ln(X)}{H}\ .
    \end{align}

    Combining the above points, we have that $\Phi_{E}(X)$ is differentiable at all $X \in \Pd{\mcl{X}}$, where
    \begin{align}
        D\Phi_{E}(X)[H] = \inner{E + \identity + \ln(X)}{H}\ .
    \end{align}
    By \cref{eq:frech-deriv-to-grad}, we may conclude $\nabla \Phi_{E}(X) = E + \identity + \ln(X)$ for $X \in \Pd{\mcl{X}}$.
\end{proof}

\section{Sublinear regret when learning positive semidefinite objects} \label{sec:subLinRegret}

In this section, we formally show that the regret function for our online learning framework scales sublinearly with the time horizon $T$.
From our generalized version of quantum Pinkser's inequality described in \cref{sec:pinsker}, we have an upper bound on $\trnorm{\omega_t - \omega_{t+1}}^2$, where $\omega_t$ is the unique strategy that follows from \cref{algo:qrftl} and is positive definite as shown in \cref{sec:pd}.
On applying \cref{cor:quantum-gen-pinsker} to bound trace distance between $\omega_t$ and $\omega_{t+1}$, we get

    \begin{equation}
        \frac{1}{4} \trnorm{\omega_{t} - \omega_{t+1}}^2 \leq \max\{\trace{\omega_{t}}, \trace{\omega_{t+1}}\} \left[ \relentr{\omega_{t}}{\omega_{t+1}} - \trace{\omega_{t} - \omega_{t+1}} \right]\ .
    \end{equation}

From the definition of Bregman divergence in \cref{eq:breg-phi-e-defn} and the gradient calculation for operator given by \cref{lem:OpDerivative}, the Bregman divergence $\bregman{\Phi_E}{\omega_t}{\omega_{t+1}}$ associated with the function $\F{E}{X} \coloneqq \inner{E}{X} + \reg{X}$ calculates to

\begin{equation}
    \bregman{\Phi_E}{\omega_t}{\omega_{t+1}} = \relentr{\omega_{t}}{\omega_{t+1}} - \trace{\omega_{t} - \omega_{t+1}} = \bregman{\reg{}}{\omega_t}{\omega_{t+1}}\ ,
\end{equation}
where our regularizer $\reg{}$ is simply the negative von Neumann entropy denoted by $-\entr{}$.

This allows us to restate the previous bound as
\begin{equation}
    \frac{1}{4} \trnorm{\omega_{t} - \omega_{t+1}}^2 \leq \max\{\trace{\omega_{t}}, \trace{\omega_{t+1}}\} \, \bregman{-\entr{}}{\omega_t}{\omega_{t+1}}\ .
\end{equation}

Next we discuss another lemma which is useful to provide to an upper bound to $\bregman{-\entr{}}{\omega_t}{\omega_{t+1}}\!$ in terms of the inner product $\inner{\nabla_t}{\omega_{t} - \omega_{t+1}}$.

\begin{lemma}\label{lemma:bregman-inner}
    \cref{algo:qrftl} guarantees that for all $t \in \{1, \ldots, T-1\}$,
    \begin{equation}
        \bregman{-\entr{}}{\omega_t}{\omega_{t+1}} \leq \eta \, \inner{\nabla_t}{\omega_{t} - \omega_{t+1}}\ .
    \end{equation}
\end{lemma}
\begin{proof}
    Let $\F{t}{X} \coloneqq \inner{E}{X} + \reg{X}$ where $\reg{} = -\entr{}$ and $E = \eta \sum_{s=1}^t \nabla_s$ (here $\nabla_s$ is as defined in \cref{algo:qrftl}).
    Then,
    \begin{subequations}\label{eq:bregman1}
    \begin{align}
        \bregman{\F{t}{}}{\omega_t}{\omega_{t+1}}
        &= \F{t}{\omega_t} - \F{t}{\omega_{t+1}} - \inner{\nabla \F{t}{\omega_{t+1}}}{\omega_t - \omega_{t+1}} \\
        &\leq \F{t}{\omega_t} - \F{t}{\omega_{t+1}}\ .
    \end{align}
    \end{subequations}
    where the first equality follows from the definition of Bregman divergence in \cref{eq:breg-phi-e-defn} and the second inequality holds due to \cref{lemma:positiveInnerProd} as the function $\F{t}{}$ is minimized at $\omega_{t+1}$ so we have the inequality $\F{t}{\omega_{t+1}} \leq \F{t}{(1-\lambda) \omega_{t+1} + \lambda \omega_t}$ for all $\lambda \in [0,1]$.
    Moreover, $\omega_{t+1}$ being positive definite (and thus in the interior of PSD cone) ensures that the $\nabla \F{t}{\omega_{t+1}}$ is well-defined.

    From our definition of $\F{t}{}$ along with the observartion that $\F{t-1}{\omega_t} \leq \F{t-1}{\omega_{t+1}}$,
    \begin{subequations}\label{eq:bregman2}
    \begin{align}
        \F{t}{\omega_t} - \F{t}{\omega_{t+1}}
        &= \F{t-1}{\omega_t} - \F{t-1}{\omega_{t+1}} + \eta \inner{\nabla_t}{\omega_t- \omega_{t+1}} \\
        &\leq \eta \, \inner{\nabla_t}{\omega_t - \omega_{t+1}}\ .
    \end{align}
    \end{subequations}
    Combining \cref{eq:bregman1} and \cref{eq:bregman2} completes the proof.
\end{proof}

\begin{proof}[Proof of \cref{lemma:InnerProductBound}]
    Combining \cref{thm:generalized-Pinsker's-inequality} and \cref{lemma:bregman-inner} we get
    \begin{equation}
        \frac{1}{4} \trnorm{\omega_t - \omega_{t+1}}^2 \leq \eta \, \max\{\trace{\omega_t}, \trace{\omega_{t+1}}\} \, \inner{\nabla_t}{\omega_t - \omega_{t+1}}\ .
    \end{equation}
    As $\inner{\nabla_t}{\omega_t - \omega_{t+1}} \leq \trnorm{\omega_t - \omega_{t+1}} \, \opnorm{\nabla_t}$ due to H\"{o}lder's inequality, we get
    \begin{equation}
        \trnorm{\omega_t - \omega_{t+1}} \leq 4 \eta \, \max\{\trace{\omega_t}, \trace{\omega_{t+1}}\} \, \opnorm{\nabla_t}\ .
    \end{equation}
    Apply H\"{o}lder's inequality again and use the above bound to obtain
    \begin{equation}
        \inner{\nabla_t}{\omega_t - \omega_{t+1}} \leq \trnorm{\omega_t - \omega_{t+1}} \, \opnorm{\nabla_t}
        \leq 4 \eta \, \max\{\trace{\omega_t}, \trace{\omega_{t+1}}\} \, \opnorm{\nabla_t}^2\ ,
    \end{equation}
    which completes the proof.
\end{proof}

\begin{proof}[Proof of \cref{thm:sublinearity}]
    Combining \cref{lemma:RegretBoundRFTL,lemma:InnerProductBound} yields
    \begin{equation}\label{eq:regret-combined}
        \regret{T} \leq 4 \eta \, \sum_{t=1}^T \left[ \max\{\trace{\omega_t},\trace{\omega_{t+1}}\} \, \opnorm{\nabla_t}^2 \right] + \frac{1}{\eta} D^2\ .
    \end{equation}
    A bound on $\opnorm{\nabla_t}$ is implied by \cref{lemma:boundOnSubgrad}.
    Alternatively, we may obtain it via the definition of $\nabla_t$, Lipschitz-continuity of $\ell_t$, and openness of $\reals$ as follows.
    \begin{equation}\label{eq:bound-nabla}
        \opnorm{\nabla_t} = \opnorm{\ell'_t(\inner{E_t}{\omega_t}) E_t} \leq B \opnorm{E_t} \leq B C\ .
    \end{equation}
    Substituting \cref{eq:bound-nabla} along with $\max\{\trace{\omega_t},\trace{\omega_{t+1}}\} \leq A$ into \cref{eq:regret-combined} yields
    \begin{equation}
        \regret{T} \leq 4 \eta \, A B^2 C^2 T + \frac{1}{\eta} D^2\ ,
    \end{equation}
    where setting $\eta = \frac{D}{2 B C \sqrt{A T}}$ attains $\regret{T} \leq 4 B C D \sqrt{A T}$, the required bound.
\end{proof}

\section*{Acknowledgements}

The authors thank
Asad Raza,
Matthias C. Caro,
Jens Eisert,
and Sumeet Khatri
for sharing a draft of their independent and concurrent work \emph{Online learning of quantum processes}.

I.G.~was supported by the National Science Foundation under Grant No.~2112890.
I.G.~also acknowledges support from an Illinois Distinguished Fellowship during a portion of this project.
A.B.~was partially supported by the BitShares Fellowship at Virginia Tech.
This research was funded in part by the Commonwealth of Virginia’s Commonwealth Cyber Initiative (CCI) under grant number 467714.
\bibliographystyle{alpha}
\bibliography{references}

\appendix

\section{Proof of Lemma~\ref{lem:first-order-taylor}}\label{appx:taylor}

In this appendix, we establish a corollary of the multivariate Taylor's theorem with the integral remainder form.
To state this theorem, we make use of the following multi-index notation.
\begin{definition}
A multi-index is an $n$-tuple of non-negative integers, i.e.,~$\alpha = (\alpha_{1},...,\alpha_{n})$ where $\alpha_{i} \in \naturals$ for all $i \in [n]$.
Given a multi-index $\alpha$ and an $n$-dimensional vector $\mbf{x}$,
\begin{align}
    |\alpha| &= \sum_{i \in [n]} \alpha_{i}\ ,&
    \alpha! &= \prod_{i \in [n]} \alpha_{i}!\ ,&
    \mbf{x}^{\alpha} &= \prod_{i \in [n]} x_{i}^{\alpha_{i}}\ ,&
    \partial^{\alpha} f &= \partial^{\alpha_{1}}_{1} \hdots \partial_{n}^{\alpha_{n}} f = \frac{\partial^{|\alpha|}f}{\partial x_{1}^{\alpha_{1}}x_{2}^{\alpha_{2}}\hdots x_{n}^{\alpha_{n}}} \ .
\end{align}
\end{definition}
With the notation stated, we present Multivariate Taylor's theorem.
We remark such a result is well-known and may be found for example in \cite{Hormander-2003a} under a slightly different parameterization.%
\begin{lemma}[Multivariate Taylor's theorem with integral remainder]
    Let $f: \reals^{n} \to \reals$ be $C^{k+1}$ on an open convex set $S$.
    If $\mbf{a} \in S$ and $\mbf{a} + \mbf{h}, \in S$ then
    \begin{equation}
        f(\mbf{a}+\mbf{h}) = \sum_{|\alpha| \leq k} \frac{\partial^{\alpha}f(\mbf{a})}{\alpha!}\mbf{h}^{\alpha} + R_{\mbf{a},k}(\mbf{h}) \ ,
    \end{equation}
    where the remainder in integral form is
    \begin{equation}
        R_{\mbf{a},k}(\mbf{h}) = (k+1) \sum_{|\alpha|=k+1} \frac{\mbf{h}^{\alpha}}{\alpha!} \int_{0}^{1} (1-t)^{k} \partial^{\alpha} f(\mbf{a}+t\mbf{h})\ .
    \end{equation}
\end{lemma}

\begin{proof}[Proof of \cref{lem:first-order-taylor}]
    Note that the $k = 0$ term is just $f(\mbf{a})$.
    The $k=1$ terms summed over is $\sum_{i} \frac{\partial f(\mbf{a})}{\partial x_{i}} h_{i} = \langle \nabla f(\mbf{a}), \mbf{h} \rangle$.
    Now we just need to deal with the remainder term,
    \begin{equation}
        R_{\mbf{a},1}(\mbf{h}) = 2 \sum_{|\alpha|=2} \frac{\mbf{h}^{\alpha}}{\alpha!} \int_{0}^{1} (1-t) \partial^{\alpha} f(\mbf{a}+t\mbf{h})\, dt\ .
    \end{equation}
    We can split this into two pieces.
    The first piece is when $\alpha$ contains a $2$ (thus the rest are zero).
    In this case, each term is of the form
    \begin{equation}
        \frac{h_{i}^{2}}{2} \int^{1}_{0} (1-t) \frac{\partial^{2}f(\mbf{a}+t\mbf{h})}{\partial x_{i}^{2}} dt \ .
    \end{equation}
    The other case is when $\alpha$ contains two ones.
    In this case, each term is of the form
    \begin{equation}
        h_{i}h_{j} \int^{1}_{0} (1-t) \frac{\partial^{2}f(\mbf{a}+t\mbf{h})}{\partial x_{i} \partial x_{j}} dt\ .
    \end{equation}
    Each such term is counted once when we sum over $\abs{\alpha} = 2$, and twice if we sum over $i \neq j$.
    We correct for this via a factor of $1/2$ when doing the latter.
    Use linearity of the integral to obtain
    \begin{subequations}
    \begin{align}
        R_{\mbf{a},1}(\mbf{h})
        &= 2 \sum_{i} \frac{h_{i}^{2}}{2} \int^{1}_{0} (1-t) \frac{\partial^{2}f(\mbf{a}+t\mbf{h})}{\partial x_{i}^{2}} dt
        + 2 \cdot \frac{1}{2} \sum_{i \neq j} h_{i}h_{j} \int^{1}_{0} (1-t) \frac{\partial^{2}f(\mbf{a}+t\mbf{h})}{\partial x_{i} \partial x_{j}} dt \\
        &= \int_0^1 (1-t) \left[ \sum_{i} h_{i}^{2} \frac{\partial^{2}f(\mbf{a}+t\mbf{h})}{\partial x_{i}^{2}}
        + \sum_{i \neq j} h_{i}h_{j} \frac{\partial^{2}f(\mbf{a}+t\mbf{h})}{\partial x_{i} \partial x_{j}} \right] dt \\
        &= \int_0^1 (1-t) \sum_{i, j} h_{i}h_{j} \frac{\partial^{2}f(\mbf{a}+t\mbf{h})}{\partial x_{i} \partial x_{j}} dt
        = \int_0^1 (1-t)\ \mbf{h}^{T} H_{f}\vert_{\mbf{a}+t\mbf{h}} \mbf{h}\ dt\ ,
    \end{align}
    \end{subequations}
    which completes the proof.
\end{proof}

\section{Proof of Lemma~\ref{lem:Carlen-Generalization}}\label{appx:proofCarlen}
The proof method is to prove the result for polynomials and then use continuity to extend it to arbitrary functions.
As such, the overall proof method is extremely similar to establishing \cite[Theorem V.3.3]{Bhatia-1997}.
\begin{proposition}\label{prop:polynomial-AtH-form}
    Let $q$ be an integer greater than or equal to one.
    Let $t \in \reals$.
    Then,
    \begin{equation}
        (A + tH)^{q} = \sum_{i = 0}^{q} t^{i} \sum_{s \in \{0,1\}^{q}: w(s) = i} \left[ \prod_{i \in [q]} \mcl{O}(s(i)) \right] \ ,
    \end{equation}
    where $w(s)$ is the Hamming weight of the string and
    \begin{equation}
        \mcl{O}(i) = \begin{cases}
            A & \text{ if } i = 0 \\
            H & \text{ if } i = 1 \ .
        \end{cases}
    \end{equation}
\end{proposition}
\begin{proof}
    We prove this by induction.\\
    \noindent \textbf{Base Case:} Let $p = 1$.
    Then,
    \begin{equation}
        (A+tH) = t^{0} \mcl{O}(0) + t^{1} \mcl{O}(1) = \sum_{i=0}^{1} t^{i} \sum_{s \in \{0,1\}: |s| = i} \left[ \prod_{i \in [1]} \mcl{O}(s(i)) \right] \ .
    \end{equation}
    \noindent \textbf{Induction Hypothesis:} For some $k \in \naturals$,
    \begin{equation}
        (A + tH)^{k} = \sum_{i = 0}^{k} t^{i} \sum_{s \in \{0,1\}^{k}: w(s) = i} \left[ \prod_{i \in [q]} \mcl{O}(s(i)) \right] \ .
    \end{equation}
    \noindent \textbf{Induction Step}
    \begin{subequations}
    \begin{align}
        & (A+tH)^{k+1} \\
        =& (A+tH)^{k}(A+tH) \\
        =& \left(\sum_{i = 0}^{k} t^{i} \sum_{s \in \{0,1\}^{k}: w(s) = i} \left[ \prod_{i \in [q]} \mcl{O}(s(i)) \right] \right) \left(A + tH \right) \\
        =& \left(\sum_{i = 0}^{k} t^{i} \sum_{s \in \{0,1\}^{k}: w(s) = i} \left[ \left\{ \prod_{i \in [q]} \mcl{O}(s(i)) \right\}A \right] \right) \\
         & \hspace{2cm} +
        \left(\sum_{i = 0}^{k} t^{i+1} \sum_{s \in \{0,1\}^{k}: w(s) = i} \left[ \left\{ \prod_{i \in [q]} \mcl{O}(s(i)) \right\}H \right] \right) \\
        =&  \sum_{i = 0}^{k+1} t^{i} \sum_{s \in \{0,1\}^{k+1}: w(s) = i} \left[ \prod_{i \in [q]} \mcl{O}(s(i)) \right] \ ,
    \end{align}
    \end{subequations}
    where the second equality is the induction hypothesis and the last equality is just noting every possible string $s \in \{0,1\}^{k+1}$ is the set of strings $s \in \{0,1\}^{k}$ with a zero appended at the end unioned with the same with one appended at the end and then re-indexing.
    This completes the proof.
\end{proof}

We make use of the following fact, which relies on similar proof methods to ones used shortly.%
\begin{lemma}\label{lem:cont-dif-frech-deriv}
    (See \cite[Theorem 3.33]{Hiai-2014a}) If $f: I \to \mathbb{R}$ is continuously differentiable over the interval $I$ and $A \in \Herm{\mcl{X}}$ has its spectrum contained in $I$, $\spec{A} \subset I$, then $f(A)$ is Fr\'{e}chet differentiable at $A$, i.e.,~$Df(A)$ exists.
\end{lemma}

\begin{proposition}\label{prop:carlen-extension-for-polynomials}
    Let $A,H \in \Herm{\mcl{X}}$.
    For all polynomials $p$,
    \begin{align} \label{eq:polynomial-trace-directional-derivative}
        D\trace{p(A)}[H] = \trace{p'(A)H} \ .
    \end{align}
\end{proposition}
\begin{proof}
    First note that the trace function is Fr\'{e}chet differentiable on $\Herm{\mcl{X}}$.
    By the chain rule for Fr\'{e}chet derivatives, $\trace{f(X)}$ is differentiable at a point $A$ so long as $f(X)$ is differentiable at point $A$.
    If $f$ is a polynomial, it is Fr\'{e}chet differentiable via Lemma \ref{lem:cont-dif-frech-deriv}.
    Thus by the statement around \cref{eq:frech-deriv-to-direc-deriv}, we may now focus on establishing
    \begin{align}\label{eq:reduc-to-direc-deriv-polynomial} \frac{d}{dt} \trace{p(A+tH)} \Big \vert_{t=0} = \trace{p'(A)H} \ .
    \end{align}
    Now note by the linearity of the derivative with respect to $t$ and the linearity of trace, it is sufficient to prove this for powers, $p_{q}(t) = t^{q}$ for $q = 1,2,3\hdots$ By Proposition \ref{prop:polynomial-AtH-form}, we have
    \begin{subequations}
    \begin{align}
        \frac{d}{dt}  p_{q}(A+tH) =& \frac{d}{dt} \sum_{i = 0}^{q} t^{i} \sum_{s \in \{0,1\}^{q}: w(s) = i} \left[ \prod_{i \in [q]} \mcl{O}(s(i)) \right] \\
        =& \sum_{i = 1}^{q} i \cdot t^{i-1} \sum_{s \in \{0,1\}^{q}: w(s) = i} \left[ \prod_{i \in [q]} \mcl{O}(s(i)) \right] \ .
    \end{align}
    \end{subequations}
    Note that when we evaluate this at $t=0$, every term but the $i=1$ term goes away.
    We thus have
    \begin{equation}
        \frac{d}{dt} p_{q}(A+tH) \Big \vert_{t=0} =  \sum_{s \in \{0,1\}^{q}: w(s) = 1} \left[ \prod_{i \in [q]} \mcl{O}(s(i)) \right] \ .
    \end{equation}
    We now take the trace of both sides.
    Because the Hamming weight of $s$ is always one, there is only one $H$ in each product.
    Thus by cyclicity of trace, $\trace{\prod_{i \in [q]} \mcl{O}(s(i))} = \trace{HA^{q-1}}$.
    Moreover, there are $\binom{q}{1} = q$ such strings.
    Thus,
    \begin{subequations}
    \begin{align}
        \frac{d}{dt} \trace{p_{q}(A+tH)} \Big \vert_{t=0} &= \sum_{s \in \{0,1\}^{q}: w(s) = 1} \trace{ \prod_{i \in [q]} \mcl{O}(s(i)) } \\
        &= q \trace{A^{q-1}H} = \trace{p_{q}'(A)H} \ ,
    \end{align}
    \end{subequations}
    where the last equality uses that $p_{q}'(x) = q p_{q-1}(x)$.
    This completes the proof.
\end{proof}

\begin{theorem}[Lemma \ref{lem:Carlen-Generalization}]
    Let $I \subset \reals$ be an open interval.
    Let $f \in C^{1}(I)$, $A \in \Herm{\mcl{X}}$ such that $\spec{A} \subset I$ and $H \in \Herm{\mcl{X}}$.
    Then,
    \begin{align}\label{eq:thm-version-of-carlen-ext}
        D\trace{f(A)}[H] = \trace{f'(A)H} \ .
    \end{align}
\end{theorem}
\begin{proof}
     We already established this result for all polynomials, so we extend it via continuity in the same manner as the proof of \cite[Theorem V.3.3]{Bhatia-1997}.
     As $f \in C^{1}(I)$, by Lemma \ref{lem:cont-dif-frech-deriv} and the chain rule for Fr\'{e}chet derivatives, the Fr\'{e}chet derivative exists for $\trace{f(A)}$ whenever $\spec{A} \subset I$.
     For notational simplicity, we denote $g(A) := \trace{f(A)}$.
     Define the RHS of \cref{eq:thm-version-of-carlen-ext} as $\mcl{D} g(A)[H]$ which is a function on the vector space of Hermitian matrices $H \in \Herm{\mcl{X}}$.

    Let $H \in \Herm{\mcl{X}}$ have norm sufficiently small that $\text{spec}(A+H) \subset I$.
    Let $[a,b] \subset I$ such that $\spec{A}, \spec{A+H} \subset [a,b]$.
    Choose a sequence of polynomials $(p_{n})$ such that $p_{n} \to f$ and $p_{n}' \to f'$ uniformly on $[a,b]$ which we know exists as we assumed $f \in C^{1}(I)$.
    Define the function $g_{n}(X) := \trace{p_{n}(X)}$.
    Let $\mcl{L}$ be the line segment joining $A$ and $A+H$ in the space of Hermitian matrices.
    Then by the mean value theorem for Fr\'{e}chet derivatives,
    \begin{subequations}
    \begin{align}
        & \|g_{m}(A+H) - g_{n}(A+H) - (g_{m}(A) -g_{n}(A))\| \\
        \leq& \|H\| \sup_{X \in \mcl{L}} \|Dg_{m}(X) - Dg_{n}(X)\| \\
        =& \|H\| \sup_{X \in \mcl{L}}\|\mcl{D} g_{m}(X) - \mcl{D} g_{n}(X)\| \ , \label{eq:Frechet-polynomial}
    \end{align}
    \end{subequations}
    where the last equality is because Proposition \ref{prop:carlen-extension-for-polynomials} showed that $D\trace{p(X)} = \mcl{D} \trace{p(X)}$ for any polynomial $p$.
    Now,
    \begin{subequations}
    \begin{align}
        \|\mcl{D} g_{m}(X) - \mcl{D} g_{n}(X)\|
        =& \sup_{\substack{H \in \Herm{X}: \\ \|H\|=1}} |\mcl{D} g_{m}(X)[H] - \mcl{D} g_{n}(X)[H]| \\
        =& \sup_{\substack{H \in \Herm{X}: \\ \|H\|=1}} \vert \trace{\left(p_{m}'(X)-p_{n}'(X)\right)H} \vert \\
        \leq& \sup_{\substack{H \in \Herm{X}: \\ \|H\|=1}} \norm{p_{m}'(X)-p_{n}'(X)}_{1} \norm{H} \\
        \leq & \norm{p_{m}'(X)-p_{n}'(X)}_{1} \ ,
    \end{align}
    \end{subequations}
    where we used the definition of $g_{n}$ and H\"{o}lder's inequality for the case $p=1$ and $q= \infty$.
    This also implies $\|\mcl{D} g_{n}(X) - \mcl{D} g(X)\| \leq \|p_{n}'(X)-f'(X)\|_{1}.$ By this bound, as we chose a sequence of polynomials such that $p_{n}' \to f'$ uniformly, it follows that for any $\varepsilon > 0$, there exists $n_{0} \in \naturals$ such that for all $n,m \geq n_{0}$
    \begin{align}\label{eq:polynomial-Cd-dif}
        \sup_{X \in \mcl{L}}\|\mcl{D} g_{m}(X) - \mcl{D} g_{n}(X)\| \leq \varepsilon/3 \ ,
    \end{align}
    and similarly
    \begin{align}\label{eq:polynomial-and-true-f-Cd-dif}
        \| \mcl{D} g(X) - \mcl{D} g_{n}(X) \| \leq \varepsilon/3 \ .
    \end{align}

    Next,
    \begin{subequations}
    \begin{align}
        & \|g(A+H) - g(A) + (g_{n}(A+H)-g_{n}(A))\| \\
        =& \lim_{m \to \infty} \|g_{m}(A+H) - g_{m}(A) + (g_{n}(A+H)-g_{n}(A))\| \\
        \leq& \|H\| \lim_{m \to \infty} \sup_{X \in \mcl{L}}\|\mcl{D} g_{m}(X) - \mcl{D} g_{n}(X)\| \\
        \leq & \frac{\varepsilon}{3} \|H\| \label{eq:polynomial-and-true-sum-dist} \ ,
    \end{align}
    \end{subequations}
    where we used \cref{eq:Frechet-polynomial} for the first inequality and \cref{eq:polynomial-Cd-dif} in the second.
    Lastly, when $\|H\|$ is sufficiently small, by the definition of Fr\'{e}chet derivative,
    \begin{align}\label{eq:Frechet-deriv-bound-for-poly}
        \|g_{n}(A+H) - g_{n}(A) -\mcl{D} g_{n}(A)\| \leq \frac{\varepsilon}{3} \|H\| \ .
    \end{align}

    Finally, we put this all together.
    That is, if $\|H\|$ sufficiently small,
    \begin{subequations}
    \begin{align}
        & \|g(A+H) - g(A) - \mcl{D} g(A)[H]\| \\
        \leq& \|g(A+H) - g(A) + (g_{n}(A+H)-g_{n}(A))\| \\
        & \hspace{1cm} + \|g_{n}(A+H) - g_{n}(A) -\mcl{D} g_{n}(A)\| + \|(\mcl{D} g(A) - \mcl{D} g_{n}(A))(H)\| \\
        \leq& \frac{\varepsilon}{3}\|H\| + \frac{\varepsilon}{3}\|H\| + \|(\mcl{D} g(A) - \mcl{D} g_{n}(A))(H)\| \\
        \leq& \frac{2\varepsilon}{3}\|H\| + \|\mcl{D} g(A) - \mcl{D} g_{n}(A)\|\|H\| \\
        \leq& \varepsilon \|H\|
        \ ,
    \end{align}
    \end{subequations}
    where the first inequality is the triangle inequality and the second is \cref{eq:polynomial-and-true-sum-dist}, \cref{eq:Frechet-deriv-bound-for-poly}, the third is just bounding the operator norm by the map's operator norm and the operator norm of $H$, and the last uses \cref{eq:polynomial-and-true-f-Cd-dif}.
    Note this means we have shown for all sufficiently small $\|H\|$ and any $\varepsilon > 0$,
    \begin{equation}
        \|g(A+H) - g(A) - \mcl{D} g(A)[H]\| \leq \varepsilon \|H\| \ .
    \end{equation}
    By definition of the Fr\'{e}chet derivative, this means $\mcl{D} g(A) = D g(A)$.
    Recalling the definition of $g$ from the beginning of the proof, this completes the proof.
\end{proof}

\section{Some useful bounds to analyze regret}

\begin{lemma}[Variational inequality]\label{lemma:positiveInnerProd}
    Let $f : S \to \reals$ be differentiable on $\interior{S}$ and $X \in \interior{S}$.
    If for a given direction $H$, there exists $\kappa > 0$ such that for all $t \in [0,\kappa]$, $f(X + tH)$ is defined and $f(X) \leq f(X + tH)$.
    Then,
    \begin{equation}
        \inner{\nabla f(X)}{H} \geq 0\ .
    \end{equation}
\end{lemma}

\begin{proof}
    As $f(X) \leq f(X + tH)$ for all $t \in [0,\kappa]$, the right directional directive of $f$ at $X$ in the direction of $H$ given by
        \begin{equation}
            \mcl{D} f(X)[H^+] \coloneqq \lim_{t \to 0^{+}} \frac{f(X + tH) - f(X)}{t} \geq 0\ .
        \end{equation}
    From \cref{prop:frech-relation}, differentiability at $X$ implies that the directional derivative of $f$ in the direction $H$ denoted by $\mcl{D} f(X)[H]$ exists with $\mcl{D} f(X)[H] = \mcl{D} f(X)[H^+] = \inner{\nabla f(X)}{H}$ which finally implies $\inner{\nabla f(X)}{H} \geq 0$.
\end{proof}

\begin{lemma}\label{lemma:boundOnSubgrad}
    Let $f : S \to \reals$ be convex and $B$-Lipschitz over an open set $S \subseteq \Herm{\mcl{X}}$, the latter meaning $\abs{f(X) - f(Y)} \leq B \trnorm{X - Y}$ for all $X,Y \in S$.
    If $\nabla_X \in \partial f_X$ is a subgradient of the function $f$ at $X \in S$, then $\opnorm{\nabla_X} \leq B$.
\end{lemma}

\begin{proof}
    From convexity of $f$, we know
    \begin{equation}
        f(Y) \leq f(X) + \inner{\nabla_X}{Y - X}, \quad \forall X,Y \in S\ .
    \end{equation}
    The Lipschitz continuity of $f$ further implies
    \begin{equation}
        B \trnorm{X - Y} \geq \abs{f(X) - f(Y)} \geq \inner{\nabla_X}{Y - X}, \quad \forall X,Y \in S\ .
    \end{equation}
    For a fixed $X \in S$, pick an $\alpha > 0$ such that $Z - X \in S$ for all $Z$ with $\trnorm{Z} \leq \alpha$.
    Note that such an $\alpha$ is guaranteed to exist by the openness of $S$.
    Set $Y = (\argmax_{\trnorm{Z} \leq \alpha} \inner{\nabla_X}{Z}) + X$, then $\trnorm{X - Y} = \alpha$ and the previous equation gets reduced to
    \begin{equation}
        \alpha B \geq \max_{\trnorm{Z} \leq \alpha} \inner{\nabla_X}{Z} = \alpha \max_{\trnorm{Z} \leq 1} \inner{\nabla_X}{Z} = \alpha \opnorm{\nabla_X}\ ,
    \end{equation}
    as the operator norm is simply dual to the trace norm.
\end{proof}

\end{document}